\documentclass[11pt]{article}
\usepackage[margin=1in]{geometry}

\usepackage{amstext}
\usepackage{amssymb}
\usepackage{amsthm}
\usepackage{amssymb}
\usepackage{amsmath}
\usepackage{epsfig}
\usepackage{float}
\usepackage{times}
\usepackage{verbatim}
\usepackage{subcaption}

\newcommand{\Pmatrix}[1]{\begin{pmatrix}#1\end{pmatrix}}
\newcommand{\theset}[1]{\left\{#1\right\}}
\newcommand{\magn}[1]{{\left| #1 \right|}}
\newcommand{\norm}[1]{\left\Vert #1 \right\Vert}
\newcommand{\st}{\textrm{\ \Big\vert\ }}
\newcommand{\tr}{^\mathrm{T}}
\newcommand{\one}{\mathbf{1}}

\newcommand{\prob}[1]{\mathbf{P}\left[#1\right]}
\newcommand{\given}{\thinspace\big\vert\thinspace}

\newcommand{\und}{\ \&\ }
\newcommand{\Det}[1]{\mathbf{det}\left[#1\right]}

\newcommand{\diag}[1]{\mathbf{diag} \left[#1\right]}
\newcommand{\bdiag}[1]{\mathbf{bdiag} \left[#1\right]}
\newcommand{\expected}[1]{ \textbf{E} \left[#1\right]}

\newtheorem{Assumption}{Assumption}[section]
\newtheorem{Proposition}{Proposition}[section]
\newtheorem{Theorem}{Theorem}[section]
\newtheorem{Example}{Example}[]
\newtheorem*{Example*}{Example}
\newtheorem{Claim}{Claim}[section]
\newtheorem{Definition}{Definition}[section]

\title{Higher-Order Uncoupled Learning Dynamics and Nash Equilibrium\thanks{
Sarah A. Toonsi (stoonsi2@illinois.edu) and Jeff S. Shamma (jshamma@illinois,edu) are with the Department of Industrial and Enterprise Systems Engineering, University of Illinois Urbana-Champaign, Urbana, Illinois, USA.}}
\author{{Sarah A. Toonsi} \and {Jeff S. Shamma}}

\begin{document}

\maketitle

\begin{abstract}
We study learnability of mixed-strategy Nash Equilibrium (NE) in general finite games using higher-order replicator dynamics as well as classes of higher-order uncoupled heterogeneous dynamics. In higher-order uncoupled learning dynamics, players have no access to utilities of opponents (uncoupled) but are allowed to use auxiliary states to further process information (higher-order). We establish a link between uncoupled learning and feedback stabilization with decentralized control. Using this association, we show that for any finite game with an isolated completely mixed-strategy NE, there exist higher-order uncoupled learning dynamics that lead (locally) to that NE.
We further establish the lack of universality of learning dynamics by linking learning to the control theoretic concept of simultaneous stabilization. We construct two games such that any higher-order dynamics that learn the completely mixed-strategy NE of one of these games can never learn the completely mixed-strategy NE of the other.
Next, motivated by imposing natural restrictions on allowable learning dynamics, we introduce the Asymptotic Best Response (ABR) property. Dynamics with the ABR property asymptotically learn a best response in environments that are asymptotically stationary. We show that the ABR property relates to an internal stability condition on higher-order learning dynamics. We provide conditions under which NE are compatible with the ABR property.  Finally, we address learnability of mixed-strategy NE in the bandit setting using a bandit version of higher-order replicator dynamics.
\end{abstract}

\maketitle

\section{Introduction}\label{sec:introduction}\ \\[-15pt]

The topic of learning in games studies dynamic learning processes in which players adapt their strategies over time in response to the evolving strategies of other players. Outcomes of learning vary from convergence to solution concepts \cite{hart2000simple,berger2005fictitious,monderer1996fictitiousplay,shamma2004unified,Piliouras2018meaning},  cycling \cite{shapley1964sometopics,foster1998nonconvergence}, stochastic stability \cite{Young1998Book}, improved social welfare \cite{Piliouras2011BeyondNE},  and even chaos \cite{piliouras2014optimization}.

Much of the literature focuses on simple-adaptive learning dynamics that can be interpreted as ``natural" behavior \cite{hart2013simple}. In particular, reference
\cite{Hart2003Uncoupled} imposes a requirement that learning dynamics are ``uncoupled", i.e., a player's dynamics do not depend explicitly on other players' utilities. The study constructs a game with a unique mixed-strategy Nash Equilibrium (NE) for which no uncoupled learning dynamics can lead to NE. The conclusion is reinforced by players being indifferent at a completely mixed-strategy NE. In particular, what motivates a player to play the exact strategy that makes opponents indifferent without access to opponents' utilities? Continuing along the lines of impossibility results, reference \cite{Papadimitriou2023Conley} shows that there are games for which any learning dynamics will have an initial condition starting from which the dynamics do not lead to NE. More results are presented in \cite{Yuzuru2002Chaos,Piliouras2019butterfly, Gkaragkounis2020Donotmix, Muthkumar2024InformsImpossible}.

This rich discussion mainly focuses on standard-order natural learning dynamics. Learning is represented as a dynamical system. Standard-order learning dynamics have a state dimension that matches the dimension of the strategy space of an agent. Higher-order learning dynamics relax this
restriction on the order so that an agent's learning dynamics can have auxiliary states. These states can have various interpretations, such as path dependencies, forecasts, or estimates.
Higher-order learning in games parallels higher-order methods in optimization \cite{daskalakis2018limit, Polyak1964momentum}, impacting areas such as machine learning \cite{daskalakis2017training,sutskever13momentum}.

In the context of learning NE, it can be the case that limitations are a direct consequence of the presumption that dynamics are of standard order. Along these lines, reference \cite{shamma2005dynamic} introduces uncoupled higher-order anticipatory learning dynamics and uses these dynamics to learn the mixed-strategy NE of the counterexample game presented by \cite{Hart2003Uncoupled}. Accordingly, the actual restriction was not that the dynamics were just uncoupled but rather that they were both uncoupled \textit{and} of standard order.

Given the disparity in conclusions, it is worthwhile to understand the limits of both what is achievable and what is impossible under higher-order learning.

Previous literature on higher-order learning focused on specific constructions of higher-order dynamics, such as \cite{basar1987relaxation, conlisk1993adaptation, flam2003newtonian, shamma2005dynamic, laraki2012higher, gao2021passivity}. Two representative examples of complementary higher-order replicator dynamics are in \cite{laraki2012higher} and \cite{arslan2006anticipatory}. Reference \cite{laraki2012higher} presents a higher-order version of replicator dynamics in which a player's dynamics can have $n$-fold aggregation of the payoff. The formulation presents a specific construction of higher-order replicator dynamics that lead to the elimination of weakly dominated strategies. As with standard-order replicator dynamics, the introduced dynamics cannot converge to mixed-strategy Nash Equilibrium (NE). Reference \cite{arslan2006anticipatory} studies anticipatory learning in general population games, introduces anticipatory replicator dynamics, and shows that anticipation can enable convergence to mixed-strategy NE.

Reference  \cite{toonsi2023Higher} shows that for any polymatrix game with an isolated, completely mixed-strategy NE, there exist uncoupled higher-order gradient play dynamics that locally lead to that NE. Furthermore, these results are robust in the sense that dynamics that lead to an equilibrium in a game continue to do so for perturbed equilibria in nearby games. The considered dynamics were expressed in a ``payoff based" structure, meaning that a player's dynamics are not explicitly dependent on any utility function, including their own. Instead, the dynamics evolve in response to an external payoff vector. When players are engaged in a game, the payoff vector becomes a function of other players' strategies. Given this formulation, the learning dynamics of a player do not change when the underlying game changes. Only the source of the payoff vector changes. A key idea enabling these results is to associate uncoupled learning to decentralized feedback stabilization in control theory.

This paper further considers higher-order learning dynamics with the initial emphasis on higher-order replicator dynamics, a widely studied class of learning dynamics that has played a major role in the literature on evolutionary games \cite{schuster1983replicator,sorin2020replicator}. We study finite games with a regularity property related to isolated equilibrium (c.f., ``regular" equilibria in \cite{Harsanyi1973Oddness}). We establish that higher-order replicator dynamics can converge to mixed-strategy NE and go on to extend the conclusions beyond higher-order replicator dynamics to various classes of higher-order uncoupled heterogeneous dynamics.

Next, we consider non-universality of learning dynamics leading to NE. By our results that show that isolated NE are learnable, we understand that there is no single negating game as in \cite{Hart2003Uncoupled}. Accordingly, we present a result about simultaneous stabilization that applies to all higher-order uncoupled learning dynamics that satisfy minimal assumptions. Given two polymatrix games, different in network topology, we show that any learning dynamics that are locally stable around the completely mixed-strategy NE of one can never be locally stable around the completely mixed-strategy NE of the other. Our construction does not specify two games explicitly. Rather, we derive a class of games leading to the desired result. Therefore, the conclusion of non-universality continues to hold for games with a unique mixed-strategy NE that match the required assumptions. These results are potentially related to learning over networks, as the network topology impacts whether convergence will occur.

Moving forward, a core question we aim to address is: ``Which NE can be learned using ``natural" dynamics?" This question requires imposing criteria for classifying natural learning dynamics, as the notion of ``natural'' is subject to interpretation. Generally, the term ``natural" is used to refer to dynamics that have two desired characteristics: simplicity and adaptivity \cite{hart2013simple}. Specific constructions of simple and adaptive natural learning dynamics include gradient play, replicator dynamics, and fictitious play. A discussion from an evolutionary perspective is presented in \cite{weibull1997evolutionary}. In that work, general selection dynamics are introduced, and criteria such as payoff monotonicity or payoff positivity are discussed. More recent work in this direction explores the connection between the system-theoretic concept of ``passivity" and the natural property of ``no-regret" \cite{cheung2021online,abdelraouf2025passivity}.

Augmenting higher-order states to standard-order learning dynamics impacts their interpretation. This issue motivates imposing restrictions on allowable higher-order dynamics and understanding their implications.
Towards this end, we introduce the Asymptotic Best Response (ABR) property for natural dynamics. Dynamics that satisfy this property play a best response (asymptotically) to any constant payoff vector. Natural dynamics such as replicator dynamics, gradient play, and fictitious play are consistent with this property. With higher-order constructions in mind, we discuss connections between ABR and internal stability of the higher-order components and show how the ABR property impacts learnability of equilibria. We continue to leverage the connection to feedback stabilization by relating the ABR property to strong stabilization in control \cite{youla1974singleloop}.

In the final section of this paper, we study learnability of mixed-strategy NE in general finite games in the bandit setting. In such a setting, players cannot access payoff vectors but only observe scalar instantaneous realized payoffs. Using the ODE method of stochastic approximation \cite{Benaim1999Dynamics,borkar2008stochastic}, we link continuous-time stabilization results to discrete-time bandit analysis of higher-order replicator dynamics.

To summarize, the specific contributions of this paper are\footnote{These results build on past work by the authors \cite{toonsi2023Higher}\cite{toonsi2024bandit}.}:
\begin{itemize}
\item Proving learnability of regular isolated completely mixed-strategy NE in general finite games using higher-order replicator dynamics or classes of higher-order uncoupled heterogeneous dynamics (Theorem~\ref{Thm:RDStabilization} and Theorem~\ref{Thm:HeterogeneousStabilization}).

\item Providing a result about non-universality of higher-order learning dynamics through the lens of simultaneous stabilization in feedback control (Theorem~\ref{Thm:SStabilization}).

\item Introducing the Asymptotic Best-Response (ABR) property, discussing its relation to internal stability of higher-order components,  and establishing a sufficient condition for higher-order learning dynamics to satisfy ABR (Proposition~\ref{Prop:Internallystable}).

\item Analyzing compatibility of NE from various classes of games with learnability via internally stable higher-order components and establishing that learning some mixed-strategy NE \textit{requires} using unstable higher-order components (Propositions~\ref{Prop:Incomptwo}, \ref{Prop:IdenticalInterest}, \ref{Prop:Zerosumgrap}, and~\ref{Prop:StrategicallyZerosum}).

\item Extending learnability results of regular isolated completely mixed-strategy NE to the bandit setting of replicator dynamics (Theorems~\ref{Thm:BanditCaseI} and~\ref{Thm:BanditCaseII}).

\end{itemize}

\section{Preliminaries}

\subsection{Finite games over mixed strategies}

Our framework will be finite games, i.e., finite players and finite actions per player. The set of players is denoted by
$$\mathcal{I}=\theset{1,\hdots,n},$$
the action set of player
$i$ is denoted by $\mathcal{A}_i=\theset{1,\hdots,k_i}$, and the utility (reward) of player $i$ is
a function
$$r_i:\mathcal{A} \rightarrow \mathbb{R},$$
where $\mathcal{A}=\mathcal{A}_1 \times \hdots \times \mathcal{A}_n.$

We now extend the definition of utility functions to mixed strategies. Define
$$\Delta(k_i) = \theset{v \in \mathbb{R}^{k_i} \st v_j\ge 0, j = 1, ..., k_i, \und \sum_{j=1}^{k_i} v_j = 1}.$$
and
$$\mathcal{X} = \Delta(k_1) \times ... \times \Delta(k_n).$$
The mixed strategy of player $i$ is a vector $x_i\in \Delta(k_i)$, and the utility of player $i$ is defined as
$$R_i : \mathcal{X} \rightarrow \mathbb{R}$$
where
$$R_i(x)=\sum_{a\in \mathcal{A}}\left(\prod_{j \in \mathcal{I}} \left(x_{j (a_{j})}\right) \right)r_i(a),$$
and $x_{j (a_{j})}$ represents the probability that player $j$ plays action $a_j$.

It will be convenient to rewrite the definition of utility to isolate the effect of an individual player.
Defining
$$\mathcal{A}_{-i} = \mathcal{A}_1 \times \dots \times \mathcal{A}_{i-1} \times
\mathcal{A}_{i+1} \times \dots \mathcal{A}_n$$
allows us to rewrite
$$r_i(a_1,...,a_i,...,a_n) = r_i(a_i,a_{-i})$$
for $a_i\in \mathcal{A}_i$ and $a_{-i}\in \mathcal{A}_{-i}$. Likewise, defining $$\mathcal{X}_{-i} = \Delta(k_1) \times ... \times \Delta(k_{i-1}) \times \Delta(k_{i+1}) \times ... \times \Delta(k_n),$$
allows us to rewrite
$$R_i(x_1,...,x_i,...,x_n) = R_i(x_i,x_{-i})$$
for $x_i\in \Delta(k_i)$ and $x_{-i}\in \mathcal{X}_{-i}$.
The utility of player $i$ can also be written as
$$R_i(x_i,x_{-i}) = x_i\tr P_i(x_{-i}),$$ where
\begin{equation}\label{eq:payoffVector}
P_i(x_{-i}) = \Pmatrix{R_i(\mathbf{e}_1,x_{-i})\\ \vdots\\ R_i(\mathbf{e}_{k_i},x_{-i})}
\end{equation}
is the \textbf{payoff vector} of player $i$,  $x_{-i} \in \mathcal{X}_{-i}$ and $\mathbf{e}_\kappa$ is the unit vector (of appropriate dimension) defined by
\begin{equation}\label{eq:DefinitionOfe}
\mathbf{e}\tr_\kappa = (0\quad\dots\quad 0\quad \underbrace{1}_{\kappa^{\mathrm{th}} \text{ position}} \quad 0\quad \dots\quad 0).
\end{equation}

The $\kappa^{\mathrm{th}}$ element of $P_i(x_{-i})$, denoted by $P_{i\kappa}$, corresponds to the expected utility of playing action $\kappa$ given strategies of other players and can be written as
\begin{equation}\label{eq:payoffVectorAlt}
P_{i\kappa}(x_{-i})= \sum_{a_{-i}\in \mathcal{A}_{-i}}\left(\prod_{j \in \mathcal{I} \backslash i} \left(x_{j (a_{j})}\right) \right)r_i(\kappa,a_{-i}),
\end{equation}
again, where $x_{j(a_j)}$ denotes the probability that player $j$ plays action $a_j$.

A \textbf{Nash Equilibrium (NE)} is a strategy profile $(x_1^*,...,x_n^*)\in \mathcal{X}$ such that for all $i \in \mathcal{I}$,
$$R_i(x_i^*,x_{-i}^*) \ge R_i(x_i,x_{-i}^*),\quad \text{ for all } x_i\in \Delta(k_i).$$
The strategy profile $(x_1^*,\hdots,x_n^*)$ is a \textbf{strict NE} if for all $i\in \mathcal{I}$, $P_i(x^*_{-i})$ has a unique maximizing element, in which case each $x^*_i$ is on a vertex of the simplex $\Delta(k_i)$.
A \textbf{completely mixed-strategy NE} is a strategy profile where all $x_i^*$ are in the interior of the simplex and for every $i\in \mathcal{I}$,
$$P_i(x_{-i}^*)=\alpha_i \one, $$
for some value, $\alpha_i$, where $\one $ is a vector of all ones. In this case, all players receive the same payoff for each possible action and, by extension,  for any mixed strategy.

\subsection{A regularity assumption}\label{sec:regularity}

This section presents a regularity assumption related to a completely mixed-strategy being isolated. See Appendix~\ref{App:A0} for an explicit example of the constructions in this section for a three-player two-action game.

Let $x^* = (x_1^*,...,x_n^*)$ be a completely mixed-strategy NE. For any $i$, one can express neighboring strategies in the simplex as
\begin{equation}\label{eq:nearby}
x_i = x_i^* + N_i w_i,
\end{equation}
where the matrix $N_i\in \mathbb{R}^{k_i \times k_{i-1}}$ satisfies
\begin{equation}\label{eq:Ncondition}
\one\tr N_i =0 \quad \And \quad   {N_i}\tr N_i =I.
\end{equation}
Note that since the NE is completely mixed, one has $x_i^* + N_i w_i \in \Delta(k_i)$
for sufficiently small $w_i$.

Now inspect the $\kappa^{\mathrm{th}}$ element of the payoff vector of player $i$, $P_i(x_{-i})$, defined in (\ref{eq:payoffVector}) and
rewritten in (\ref{eq:payoffVectorAlt}) in the vicinity of $x^*$:
$$P_{i\kappa}(x_{-i}^*+N_{-i}w_{-i}) =\sum_{a_{-i}\in \mathcal{A}_{-i}}\left(\prod_{j \in \mathcal{I} \backslash i} \left(x^*_{j (a_{j})}+ N_j(a_{j}) w_{j}\right) \right)r_i(\kappa,a_{-i}),$$
where
$$N_{-i}w_{-i}=\Pmatrix{N_{1}w_{1} & \dots & N_{i-1}w_{i-1} & N_{i+1}w_{i+1} & \dots & N_{n}w_{n}}.$$
The notation $N_i(a_j)$ refers to the row of the matrix $N_i$ corresponding to action $a_j$.

Collecting relevant terms together, the payoff vector of player $i$ near $x^*$ takes the form
\begin{equation}\label{eq:payoffnearxstar}
P_{i}(x_{-i}^*+N_{-i}w_{-i})=P_{i}(x_{-i}^*)+M_{i} \mathcal{N} w +\Tilde{P}_{i}(w_{-i}),
\end{equation}
where
$$M_{i}= \Pmatrix{M_{i1}& \hdots & M_{i(i-1)}& 0& M_{i(i+1)}&\hdots&M_{in}},$$
$$M_{ij} = \nabla_{x_j} P_i(x^*),\quad i\not= j,$$
\begin{equation}
\mathcal{N} = \bdiag{N_1,...,N_n},
\end{equation}
and $\bdiag{\cdot}$ is defined as a block diagonal matrix formed from its arguments.
The elements of $\Tilde{P}_{i}(w_{-i})$ are polynomials with terms at least quadratic in $w_{-i}$ so that
$$\limsup_{\norm{w_{-i}}\downarrow 0} \frac{\norm{\Tilde{P}_i(w_{-i})}}{\norm{w_{-i}}^2} < \infty.$$

The \textbf{reduced-order payoff vector} is defined as
$$N_i\tr P_{i}(x_{-i}^*+N_{-i}w_{-i})= \underbrace{N_i\tr M_{i}\mathcal{N}}_{\mathcal{M}_{i}} w +N_i\tr \Tilde{P}_{i}(w_{-i}),$$
where $\mathcal{M}_i$ is the reduced-order version of $M_i$.
Let
\begin{equation}\label{eq:calligraphicM}
M=\Pmatrix {M_{1}\\ \vdots \\ M_{n}} \quad \And \quad
\mathcal{M}=\mathcal{N}\tr  M \mathcal{N}=\Pmatrix{\mathcal{M}_{1}\\ \vdots \\ \mathcal{M}_{n}}.
\end{equation}
In brief, the matrix $\mathcal{M}$ is the Jacobian matrix, evaluated at $0$, of the mapping
$$w \mapsto \mathcal{N}\tr \Pmatrix{P_1(x_{-1}^* + N_{-1}w_{-1})\\ \vdots\\ P_n(x_{-n}^* + N_{-n} w_{-n})}$$
i.e., the mapping from perturbations of the mixed-strategy NE to the reduced-order payoff vector.

We will make the following assumption on completely mixed-strategy Nash equilibria.
\begin{Assumption}\label{Assumption:Regularity}
The (Jacobian) matrix  $\mathcal{M}$ is non-singular.
\end{Assumption}

The intuition behind this assumption is to require a certain regularity in the games such that a local differentiable inverse map exists from reduced-order payoff vectors to perturbations around completely mixed-strategy NE of interest. Accordingly, the completely mixed NE must be isolated. By the inverse function theorem \cite[p. 39]{spivak1965calculus}, the regularity assumption on the local inverse function necessitates non-singularity of $\mathcal{M}$. This condition is closely related to the definition of regular equilibria introduced in \cite{Harsanyi1973Oddness}, where it was shown that almost all finite games are regular (a game is called regular if all its equilibria are regular). Further discussion on uniqueness of equilibria can be found in \cite{bailey2024uniqueness}.

\section{Payoff-based learning dynamics: Standard and higher-order}\ \\[-15pt]

An important aspect of the approach in the paper is that a player's learning dynamics are defined as an \textit{open dynamical system} that maps an input vector to a mixed strategy.

In \textit{standard-order} learning dynamics, these dynamics for an individual player $i$ take the form
\begin{equation}\label{eq:standardorder}
\dot{x}_i(t)=f_i(x_i(t),p_i(t)),
\end{equation}
where $f_i:\Delta(k_i) \times \mathbb{R}^{k_i} \rightarrow  \mathbb{R}^{k_i}$ and $p_i:\mathbb{R}_{+} \rightarrow \mathbb {R}^{k_i}$. These dynamics are ``standard-order'' in the sense that the evolving state variable is the player's mixed strategy.

We assume implicitly that both $f_i$ and $p_i$ are such that (\ref{eq:standardorder}) has unique solutions when $x_i(0) \in \Delta (k_i)$ and that $p_i$ is continuous in $t$. Moreover, we assume that the simplex is invariant, i.e., if $x_i(0) \in \Delta (k_i)$, then
$ x_i(t) \in \Delta(k_i)$ for all $t \geq 0.$

The external input vector, $p_i(t)$, becomes a payoff vector once the dynamics are in feedback with a game, i.e.,
$$p_i(t)=P_i(x_{-i}(t)).$$ This payoff formulation allows defining learning dynamics independent of any game and facilitates the use of feedback analysis tools (see Figure~\ref{fig:feedback}). Accordingly, a player's learning dynamics are not explicitly dependent on any utility function, including their own. This feature renders the dynamics \textit{uncoupled} in the terminology of \cite{Hart2003Uncoupled}.
\begin{figure}[!t]
\centerline{\includegraphics[width=.25\columnwidth]{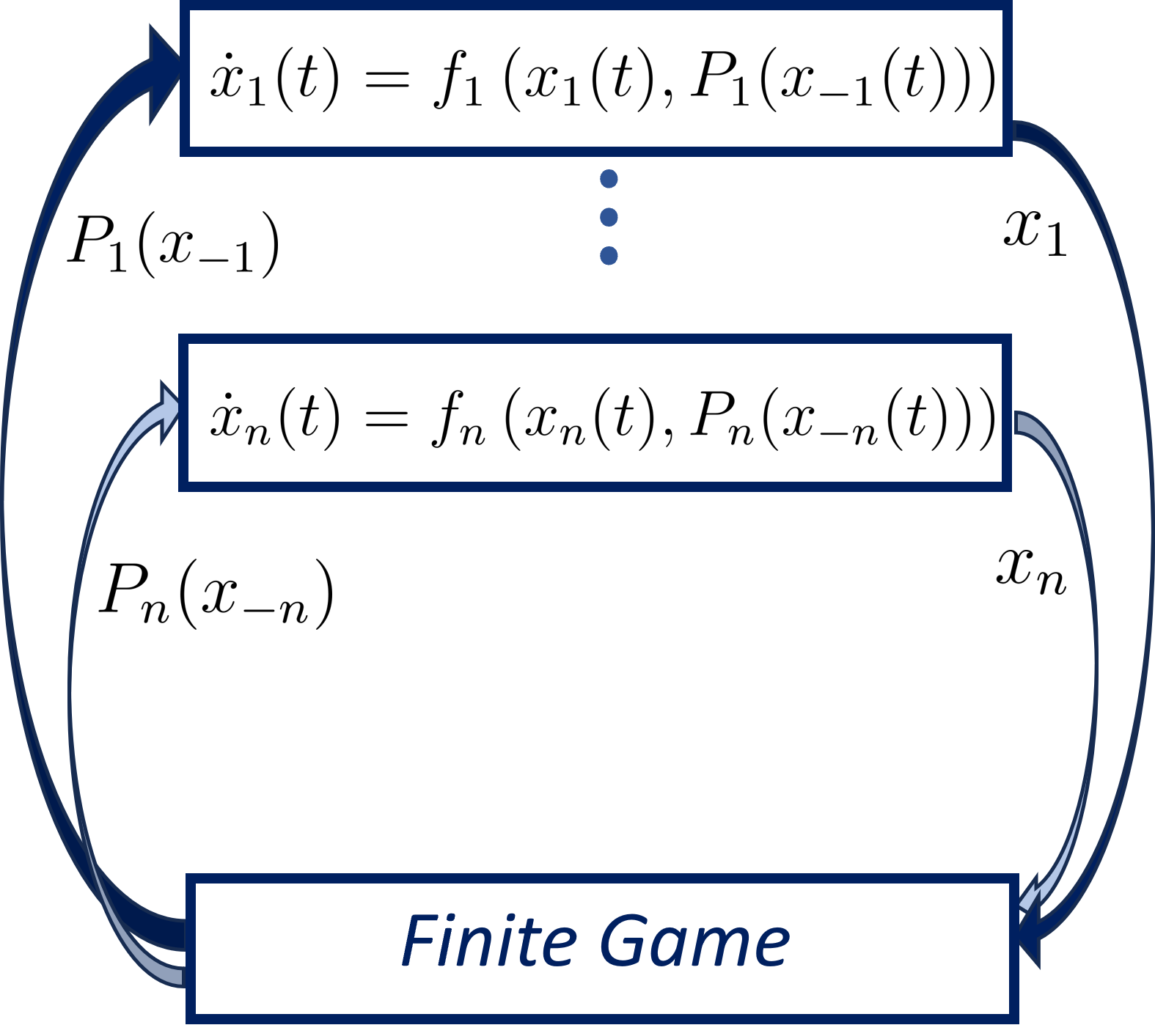}}
\caption{Learning dynamics in feedback with a game}
\label{fig:feedback}
\end{figure}

In higher-order learning, players are allowed auxiliary internal states to further process the payoff vector. There is no longer any restriction to the order of the dynamics of a player.

For any standard-order dynamics, we can define the following \textit{higher-order learning dynamics}:
\begin{subequations}\label{eq:ourhigherorder}
\begin{align}
\dot{x}_i(t)&=f_i(x_i(t),p_i(t)+\phi_i(p_i(t), z_i(t)))\\
\dot{z}_i(t)&=g_i(p_i(t),z_i(t)),
\end{align}
\end{subequations}
where $x_i, p_i,$ and $f_i$ are as before. We  assume implicitly that $f_i,g_i,p_i,$ and $\phi_i$ are such that the dynamics have unique solutions whenever $x_i(0)\in \Delta(k_i)$.

The auxiliary state variable here is $z_i \in \mathbb{R}^{d_i}$, which evolves according to $g_i:\mathbb{R}^{k_i} \times \mathbb{R}^{d_i} \rightarrow  \mathbb{R}^{d_i}$. The higher-order terms also act as an open system that maps an input vector to the auxiliary state.

In this form of additive higher-order learning,  player $i$ continues to utilize standard-order dynamics, but in response to the modified vector $p_i + \phi_i$ instead of just $p_i$, where $\phi_i:\mathbb{R}^{k_i} \times \mathbb{R}^{d_i} \rightarrow  \mathbb{R}^{k_i}$. The modified payoff vector $p_i + \phi_i$ can capture more complex behaviors, such as path dependencies, anticipation, or estimation. It is noted that more general constructions of higher-order dynamics could have non-additive effect of $\phi_i$ on the function $f_i$. Furthermore, both $g_i$ and $\phi_i$ could also be functions of the strategy $x_i$. We will not consider such generalizations herein.

We will impose the following assumption on the higher-order modifications.

\begin{Assumption}\label{Assumption:higherorder}
For any $p_i^*$ there exists $z_i^*$ such that
\begin{align*}
0&=g_i(p^*_i,z^*_i)\\
0&=\phi_i(p^*_i,z^*_i).
\end{align*}
\end{Assumption}
This assumption assures that stationary points of standard-order learning corresponding to NE continue to be stationary points under higher-order learning.

\section{Learning mixed-strategy NE with higher-order replicator dynamics} \label{Sec:Stabilization}\ \\[-15pt]

In this section, we establish learnability of mixed-strategy NE in finite games using higher-order replicator dynamics.

\subsection{Standard and higher-order replicator dynamics}\label{sec:standard-order}

Following (\ref{eq:standardorder}), our starting point is to express replicator dynamics (cf., \cite{sorin2020replicator}) as an open system. For player $i$, these become
\begin{equation}\label{eq:standardOrderReplicator}
\dot{x}_i=\diag{p_{i}-\left(  x_{i} \tr p_{i} \right) \one}x_{i},
\end{equation}
where $\diag{v}$ denotes a diagonal matrix of the elements of the argument vector, $v$.
When in feedback with a game, the dynamics of player $i$ take the form
\begin{equation}\label{eq:RDclosed}
\dot{x}_i=\diag{P_{i}(x_{-i})- \left( x_{i} \tr  P_{i}(x_{-i}) \right)  \one}x_{i},
\end{equation}
which is now in the form of a closed system, where the $x_{-i}$ also evolve according to replicator dynamics.

Our structure for higher-order dynamics in (\ref{eq:ourhigherorder}) applied to replicator dynamics  becomes
\begin{subequations}\label{eq:higherOrderRDgeneral}
\begin{align}
\dot{x}_{i}&=\diag{p_{i}+\phi_{i}(p_i,z_i) -\left( x_{i}\tr (p_{i}+\phi_{i}(p_i,z_i)) \right)\one}x_{i}\\
\dot{z}_i&=g_i(p_i,z_i).
\end{align}
\end{subequations}
As before, these are written as an open system dependent on the external input $p_i$.

\subsection{Examples}

Here, we compare and contrast different forms of higher-order replicator dynamics that have appeared in the literature. Let us first revisit a derivation of replicator dynamics \cite{Hofbauer2009DerivationofRD}. 
Consider a score variable $s_i$ that has the dynamics
$$\dot{s}_i= p_i,$$
i.e., $s_i$ is the integral of the external input vector, viewing replicator dynamics as an open system. Next, define the strategy of player $i$ as
$$x_i=\sigma(s_i),$$
where $\sigma(\cdot)$ is the softmax function (or Gibbs distribution)
$$\sigma: \mathbb{R}^\kappa \rightarrow \Delta(\kappa)$$
defined as
$$\sigma(v) = \frac{1}{\sum_{j=1}^\kappa e^{v_j}} \Pmatrix{e^{v_1}\\ \vdots\\ e^{v_\kappa}},$$ where
$v = (v_1,...,v_\kappa)$. It is straightforward to show that differentiating $x_i$ results in standard-order replicator dynamics (\ref{eq:standardOrderReplicator}).

Reference \cite{gao2021passivity} introduces ``exponentially discounted learning'':
\begin{subequations}\label{eq:EDL}
\begin{align}
\dot{s}_i &= p_i - s_i\\
x_i &= \sigma(s_i)
\end{align}
\end{subequations}
which can be written to mirror (\ref{eq:higherOrderRDgeneral}) as
\begin{align*}
\dot{x}_i &= \diag{(p_i - s_i) - \left(x_i\tr (p_i - s_i)\right)\one }x_i\\
\dot{s}_i &= p_i - s_i.
\end{align*}
It is noted that the original form in (\ref{eq:EDL}) is standard-order, but does not have the strategy as a state variable. Reference \cite{gao2021passivity} goes on to allow additional higher-order behaviors in the form of dynamic dependence on the strategy. A representative special case is
\begin{subequations}\label{eq:EDLfeedback}
\begin{align}
\dot{x}_i &= \diag{(p_i - \xi_i - s_i) - \left(x_i\tr (p_i - \xi_i - s_i)\right)\one }x_i\\
\dot{s}_i &= p_i - \xi_i - s_i\\
\dot{\xi}_i &= x_i - \xi_i.
\end{align}
\end{subequations}

Reference \cite{arslan2006anticipatory} introduces ``anticipatory'' replicator dynamics, which take the form:
\begin{subequations}\label{eqAnticipatory}
\begin{align*}
\dot{x}_i &= \diag{(p_i + \gamma\lambda(p_i - z_i)) - \left(x_i\tr \left(p_i + \gamma\lambda(p_i - z_i)\right)\right)\one}x_i\\
\dot{z}_i &= \lambda(p_i - z_i).
\end{align*}
\end{subequations}
The terminology stems from the quantity $p_i + \gamma\lambda(p_i - z_i)$ serving as a myopic prediction of the future value of $p_i$ for $\lambda \gg 1$. Reference \cite{arslan2006anticipatory} shows that this modification can enable convergence to mixed-strategy NE in cases where standard-order replicator dynamics are unstable.

Reference \cite{laraki2012higher} introduced higher-order replicator dynamics that parallel the original, except that the ``score'' variable is a higher-order integral of the input vector, as in
\begin{align*}
\frac{d^m s_i}{dt^m} &= p_i\\
x_i &= \sigma(s_i).
\end{align*}
In case $m = 2$, these dynamics can be written as
\begin{subequations}\label{eq:2DRD}
\begin{align}
\dot{x}_i &= \diag{z_i - \left(x_i\tr z_i\right) \one }x_i\\
\dot{z}_i &= p_i.
\end{align}
\end{subequations}
Reference \cite{laraki2012higher} goes on to show that this form of higher-order replicator dynamics can lead to behaviors unachievable by standard-order replicator dynamics, such as the elimination of weakly dominated strategies.

These examples illustrate the breadth of possibilities in defining higher-order learning dynamics. In particular, it is noted that both (\ref{eq:EDLfeedback}) and (\ref{eq:2DRD}) do not fall within the present framework of (\ref{eq:higherOrderRDgeneral}), but for different reasons, (\ref{eq:EDLfeedback}) because of the presence of $x_i$ in the higher-order state dynamics, and (\ref{eq:2DRD}) because of violation of Assumption~\ref{Assumption:higherorder}.

\subsection{Stabilizability of isolated completely mixed-strategy NE}

We now return to the existence of higher-order replicator dynamics that lead to mixed-strategy NE. We will focus on the following special case of higher-order replicator dynamics, where the higher-order dynamics and their additive effect
on the payoff vector are defined by a linear dynamical system:
\begin{subequations}\label{eq:higherorderRD}
\begin{align}
\dot{x}_{i}&=\diag{p_{i}+\phi_{i}(p_i,v_i,\xi_i)-\left( x_{i}\tr \left (p_{i}+\phi_{i}(p_i,v_i,\xi_i)\right) \right)\one}x_{i}\\
\dot{v}_i&= N_i\tr p_i-v_i, \label{eq:higherorderRD3}\\
\dot{\xi}_i&=E_i\xi_i+F_i(N_i\tr p_i-v_i)\\
&\quad \phi_{i}(p_i,v_i,\xi_i) =N_i( G_{i} \xi_{i}+H_{i} (N_i\tr p_i-v_i))
\end{align}
\end{subequations}
In this formulation, the auxiliary state variable $z_i$ consists of two components, $v_i$ and $\xi_i$. The matrices $(E_i,F_i,G_i,H_i)$ define the dynamics of the auxiliary states and how they impact the standard-order replicator dynamics. The $v_i$ dynamics represent a washout filter (see Appendix~\ref{App:WashoutFilters}) that has the property that if $p_i(t)$ converges to a constant, then $(N_i\tr p_i(t)-v_i(t))$ converges to zero. The washout filter guarantees satisfaction of Assumption~\ref{Assumption:higherorder}. In particular, for any $p_i^*$, setting
$v_i^* = N_i\tr p_i^*$ and $\xi_i^* = 0$ result in
$$\dot{v}_i = 0, \quad \dot{\xi}_i = 0, \quad \&\quad \phi_i(p_i^*,v^*_i,\xi_i^*) = 0,$$
as desired.
 In fact, for $p_i^* = \alpha_i \one $, one has that $v_i^* = 0$ by definition of $N_i$. Figure~\ref{fig:openloophigherorder} illustrates the structure of these dynamics.

\begin{figure}[!t]
\centerline{\includegraphics[width=.8\columnwidth]{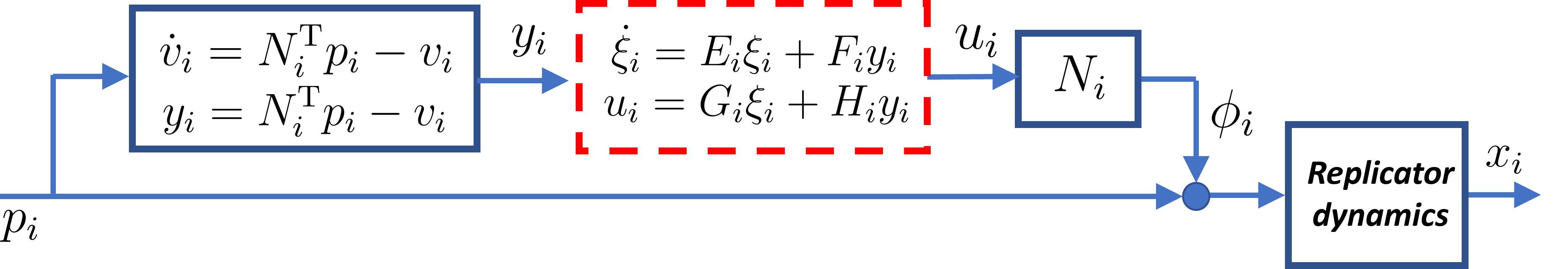}}
\caption{Higher-order replicator dynamics with linear higher-order terms as an open system}
\label{fig:openloophigherorder}
\end{figure}

Assume all players utilize higher-order replicator dynamics as in (\ref{eq:higherorderRD}).
If $x^* = (x_1^*,...,x_n^*)$ is a mixed-strategy equilibrium and $p_i^* = \alpha_i^*\one $ for all $i$, then
\begin{align*}
x^* &= (x_1^*,...,x_n^*)\\
v^* &= (v_1^*,...,v_n^*) = (0,...,0)\\
\xi^* &= (\xi_1^*,...,\xi_n^*) = (0,...,0)\\
\end{align*}
is a stationary point of the overall dynamics.

We are now in a position to state the first main result.

\begin{Theorem}\label{Thm:RDStabilization}
For any finite game with an isolated completely mixed-strategy NE, $x^*$, that satisfies Assumption~\ref{Assumption:Regularity}, there exist higher-order replicator dynamics of the form (\ref{eq:higherorderRD}) with $p_i = P_i(x_{-i})$ such that the equilibrium $(x^*,0,0)$ is locally exponentially stable.

Furthermore, there exists a $\delta > 0$, such that stability is maintained for nearby games, i.e., for all $\tilde{r}:\mathcal{A}\rightarrow \mathbb{R}^n$ such that
$$\max_{a\in\mathcal{A}} \norm{r(a) - \tilde{r}(a)} < \delta.$$
\end{Theorem}

The remainder of this section is devoted to the proof of Theorem~\ref{Thm:RDStabilization}.

\subsection{Reformulation as decentralized stabilization}

Theorem~\ref{Thm:RDStabilization} claims local stability of NE is achievable under higher-order replicator dynamics. The proof will rely on a control theoretic approach that uses linear systems techniques to study the local behavior of nonlinear systems (see Appendix~\ref{sec:ofStabilization}).

Let us rewrite our specific form of higher-order replicator dynamics (\ref{eq:higherorderRD}) as
\begin{subequations}\label{eq:closedloopHORD}
\begin{align}
\dot{x}_{i}&=\diag{P_{i}(x_{-i})+N_i u_i -\left( x_{i}\tr (P_{i}(x_{-i})+N_i u_i) \right)\one}x_{i}\\
\dot{v}_i&=N_i\tr P_{i}(x_{-i}) -v_i   \\
y_i&=N_i\tr P_{i}(x_{-i}) -v_i
\end{align}
\end{subequations}
where we have replaced $p_i$ with $P_i(x_{-i})$ and introduced two variables, a ``measurement'', $y_i$, and ``control'', $u_i$. The control, $u_i$, is interpreted as the additive payoff modification, which evolves according to a linear dynamical system driven by the measurement, $y_i$, according to
\begin{subequations}\label{eq:controller}
\begin{align}
\dot{\xi_i} &= E_i \xi_i + F_i y_i\\
u_i &= G_i \xi_i + H_i y_i.
\end{align}
\end{subequations}

We are now in the setting of Appendix~\ref{sec:ofStabilization}. Let us gather the state variables
$x = (x_1,...,x_n)$, $v = (v_1, ..., v_n)$, measurement variables, $y = (y_1,...,y_n)$, and control variables, $u = (u_1,...,u_n)$.
The overall system can be written as
\begin{subequations}\label{eq:overallRD}
\begin{align}
\Pmatrix{\dot{x}\\ \dot{v}} &= \Pmatrix{\mathcal{F}(x,v,u)\\ \mathcal{N}\tr P(x) - v}    \\
y &= \mathcal{N}\tr P(x) - v
\end{align}
\end{subequations}
where
$$P(x) = \Pmatrix{P_1(x_{-1})\\ \vdots\\ P_n(x_{-n})}.$$
These dynamics have an equilibrium $(x^*,v^*,u^*) = ((x_1^*,...,x_n^*),(0,...,0),(0,...,0)).$

Our goal is to find linear dynamic feedback
\begin{align*}
\dot{\xi} &= E\xi + Fu\\
u &= G\xi + Hy
\end{align*}
so that the overall dynamics
\begin{subequations}\label{eq:RDfeedback}
\begin{align}
\Pmatrix{\dot{x}\\ \dot{v}} &= \Pmatrix{\mathcal{F}\Big(x,v,G\xi+H(\mathcal{N}\tr P(x) - v)\Big)\\ \mathcal{N}\tr P(x) - v}\\
\dot{\xi} &= E \xi + F(\mathcal{N}\tr P(x) - v)
\end{align}
\end{subequations}
are locally exponentially stable at $(x^*,0,0)$.

An important \textit{distinction} with the stabilization framework of Appendix~\ref{sec:ofStabilization} is that the feedback must remain ``uncoupled'' to respect that players only have access to their own payoff vector. This restriction necessitates the block diagonal structure
\begin{align*}
E &= \bdiag{E_1,...,E_n}, \quad F = \bdiag{F_1,...,F_n},\\
G &= \bdiag{G_1,...,G_n}, \quad H = \bdiag{H_1,...,H_n}.
\end{align*}

This restriction brings in the discussion of Appendix~\ref{App:Decentralized}, which provides conditions under which decentralized (i.e., or in our case, uncoupled) stabilization is possible. The remainder of the proof will require deriving the associated linearized dynamics of (\ref{eq:overallRD}) (see upcoming (\ref{eq:DecentralizedFormulation})) and showing that the decentralized stabilization conditions in Appendix~\ref{App:Decentralized} are satisfied.

\subsection{Proof of Theorem~\ref{Thm:RDStabilization}}

Following the process of Appendix~\ref{sec:ofStabilization}, we first construct the linearized dynamics of (\ref{eq:overallRD}) at the equilibrium $(x^*,0,0)$.
To this end, define the block matrix, $J$, with block elements
$$J_{ij} = \nabla_{x_j} \mathcal{F}_i (x^*,0,0)$$
where, following (\ref{eq:closedloopHORD}),
$$\mathcal{F}_i(x,v,u) = \diag{P_{i}(x_{-i})+N_i u_i -\left( x_{i}\tr (P_{i}(x_{-i})+N_i u_i) \right)\one}x_{i}.$$
For $i=j$, one can show that
$$\nabla_{x_i} \mathcal{F}_i(x,v,u) = \diag{P_{i}(x_{-i})+N_i u_i -\left( x_{i}\tr (P_{i}(x_{-i})+N_i u_i) \right)\one} - \diag{x_i}\one (P_i(x_{-i}) + N_i u_i)\tr.$$
Evaluating this expression at $x = x^*$, $v=0$, and $u=0$, results in
\begin{align*}
J_{ii} &= \diag{P_{i}(x^*_{-i}) -\left( (x^*_{i})\tr (P_{i}(x^*_{-i})) \right)\one}-\diag{x^*_i}\one (P^*_i(x_i))\tr\\
&= \diag{\alpha_i^* \one  - \left( (x^*_{i})\tr \alpha_i^*\one \right)\one} - \alpha_i^*\diag{x_i^*}\one \one\tr\\
&= -\alpha_i^* x_i^*\one\tr
\end{align*}
where we used that at a completely mixed-strategy NE, $P_i^*(x_i^*) = \alpha_i^*\one$ for some $\alpha_i^*$.

For $i\not= j$,
\begin{equation}\label{eq:nondiagonalblockofJ}
\nabla_{x_j} \mathcal{F}_i(x,v,u) =
\diag{x_i} \left( \nabla_{x_j} P_{i}(x_{-i}) - \one \Big( x_{i}\tr \big(\nabla_{x_j} P_{i}(x_{-i})\big)\Big)\right) .
\end{equation}
Evaluating this expression at $x=x^*$, $v=0$, and $u=0$, results in
$$J_{ij} = \diag{x_i^*}\left(M_{ij} - \one (x_i^*)\tr M_{ij}\right).$$
Using similar arguments,
$$\nabla_{u_i} \mathcal{F}_i(x,v,u) =
\diag{x_i}\left(N_i - \mathbf{1} x_i\tr N_i\right).$$
Define
\begin{align*}
B_i &= \nabla_{u_i} \mathcal{F}_i(x^*,0,0)\\
&= \diag{x_i^*}\left(N_i - \mathbf{1} (x_i^*)\tr N_i\right).
\end{align*}

Let $\tilde{x} = x - x^*$. The linearized dynamics of (\ref{eq:closedloopHORD}) now can be written as
\begin{align*}
\Pmatrix{\dot{\tilde{x}}\\ \dot{v}} &= \Pmatrix{J&0\\ \mathcal{N}\tr M&-I}\Pmatrix{\tilde{x}\\ v} + \Pmatrix{\bdiag{B_1,...,B_n}\\0} u\\
y &= \Pmatrix{\mathcal{N}\tr M& -I} \Pmatrix{\tilde{x}\\ v}
\end{align*}

Recall equation (\ref{eq:nearby}) in Section~\ref{sec:regularity}. To study the local behavior around $x^*$,  we  define
$$\tilde{x}_i= x_i - x_i^* = N_iw_i$$
for each player $i$, or collectively
$$\tilde{x} = \mathcal{N} w.$$
The collective reduced-order linearized dynamics of the open system now take the form
\begin{subequations}\label{eq:DecentralizedFormulation}
\begin{align}
\Pmatrix{\dot{w} \\ \dot{v}}
&= \Pmatrix{\mathcal{N}\tr J\mathcal{N}&0\\\mathcal{M} & -I}\Pmatrix{w\\ v} + \Pmatrix{\bdiag{N_1\tr B_1,...,N_n\tr B_n}\\0}u\\
y &= \Pmatrix{\mathcal{M}&-I}\Pmatrix{w\\ v}
\end{align}
\end{subequations}
which uses that $\mathcal{M} = \mathcal{N}\tr M\mathcal{N}$.

Before proceeding, it is interesting to note that linearized standard-order replicator dynamics near a completely mixed equilibrium satisfy
$$\dot{w} = \mathcal{N}\tr J \mathcal{N} w.$$
The block diagonal terms are
$$N_i\tr J_{ii} N_i = -\alpha_i^* N_i\tr x_i^* \one\tr N_i = 0$$
by definition of $N_i$. Accordingly, the matrix $\mathcal{N}\tr J\mathcal{N}$ has zero trace, which implies that a completely mixed NE cannot be locally exponentially stable under standard-order replicator dynamics.

We now show that the conditions of Appendix~\ref{App:Decentralized} are satisfied by the linearized dynamics (\ref{eq:DecentralizedFormulation}), implying the existence of decentralized stabilizing feedback, thereby completing the proof of
Theorem~\ref{Thm:RDStabilization}.

\begin{Claim}\label{Claim:nonSingular}
Under Assumption~\ref{Assumption:Regularity}, $\mathcal{N} \tr J \mathcal{N}$ is non-singular.
\end{Claim}

\begin{proof} We will show that there is no non-zero vector such that
$$\mathcal{N}\tr J \mathcal{N} c = 0.$$
Let $c = (c_1,...,c_n)$ be suitably partitioned and suppose the $i^\mathrm{th}$ block component of $\mathcal{N}\tr J\mathcal{N}$
satisfies
$$N_i\tr \diag{x_i^*} \underbrace{\sum_{j\not= i} (M_{ij} N_j - \mathbf{1}(x_i^*)\tr M_{ij}N_j) c_j}_{\zeta_i} = 0.$$
First, note that $\diag{x_i^*} \zeta_i$ is orthogonal to the null space of $N_i\tr$, which is spanned by $\one$:
\begin{align*}
\one\tr \diag{x_i^*}  \zeta_i &= \one\tr \diag{x_i^*} \sum_{j\not= i} (M_{ij} N_j - \mathbf{1}(x_i^*)\tr M_{ij}N_j)\\
&= 0.
\end{align*}
Accordingly, since $\diag{x_i^*}$ is invertible, it must be that $\zeta_i = 0$, or
$$\sum_{j\not= i} M_{ij}N_j c_j = \sum_{j\not= i} \mathbf{1}(x_i^*)\tr M_{ij}N_j c_j.$$
Multiplying on the left by $N_i\tr$ results in
$$\sum_{j\not=i} N_i\tr M_{ij} N_j c_j = \sum_{j\not= i} N_i\tr \one (x_i^*)\tr M_{ij}N_j c_j = 0,$$
and therefore
$$\mathcal{M}_i c = 0.$$
Repeating this analysis for all block components results in
$$\mathcal{M}c = 0.$$
Since $\mathcal{M}$ is non-singular by assumption, $c= 0$.
\end{proof}

\begin{Claim}\label{Claim:Birank}
The $k_i - 1\times k_i - 1$ matrix
$$N_i\tr B_i = N_i\tr \diag{x_i^*}(N_i - \one (x_i^*)\tr N_i)$$
is non-singular.
\end{Claim}

\begin{proof} Following arguments similar to the proof of Claim~\ref{Claim:nonSingular}, one can show that
if
$$N_i\tr B_i c = N_i\tr \diag{x_i^*}(N_i - \one (x_i^*)\tr N_i) c = 0$$
for some $c$, then necessarily
$$(N_i - \one (x_i^*)\tr N_i) c= 0.$$
Multiplying on the left by $N_i\tr$ implies that $c = 0$.
\end{proof}

\begin{Claim}\label{Claim:final}
For all partitions, $\mathcal{U} \cup \mathcal{Y} = \theset{1,2,...,n}$, the rank of partitioned matrix
\begin{equation}\label{eq:DecentralizedtTestMatrix}
\mathbf{\Sigma} = \Pmatrix{\Pmatrix{\mathcal{N}\tr J \mathcal{N} -\lambda I & 0\\ \mathcal{M}&-(\lambda+1)I} &
\left.\Pmatrix{\bdiag{N_1\tr B_1,...,N_n\tr B_n}\\0}\right\vert^\mathcal{U}\\
\left.\Pmatrix{\mathcal{M} & \hbox{\hspace{50pt}}& -I}\right\vert_\mathcal{Y} & 0 } \end{equation}
is $2\ell$ for all $\lambda$ with positive real-part , where $$\ell=\sum_{i=1}^{n} (k_i -1).$$
\end{Claim}

\begin{proof}
 First, note that $\mathbf{\Sigma}$ has rank $2\ell$ for $\lambda =0$ from the top left block matrix by non-singularity of $\mathcal{N}\tr J\mathcal{N}$ from Claim~\ref{Claim:nonSingular}.

Perform the following block row operations: (i) multiply the bottom block row by $-(\lambda+1)$, (ii) add the rows corresponding to $\mathcal{Y}$ from the middle block row to the bottom block row, and (iii) swap the first and last block rows. Then the rank of $\mathbf{\Sigma}$ is the rank of
$$\Pmatrix{-\lambda \mathcal{M}\vert_\mathcal{Y} & 0 & 0 \\ \mathcal{M} & -(\lambda+1) I & 0\\\mathcal{N} \tr J \mathcal{N} - \lambda I& 0&\left.\Pmatrix{\bdiag{N_1\tr B_1,...,N_n\tr B_n}}\right\vert^\mathcal{U}}.\nonumber$$
The first block row provide a row rank of $\sum_{e \in \mathcal{Y}} (k_e -1)$ since $\mathcal{M}$ is non-singular and $\lambda \neq 0$. The middle block row provides a row rank of $\ell$. Finally, the last block row provides a row rank of $\sum_{q\in \mathcal{U}} (k_q -1)$ , since each matrix $N_i\tr {B}_i$ has rank $k_i-1$.

\end{proof}

With Claim~\ref{Claim:final}, the proof of Theorem~\ref{Thm:RDStabilization} is complete. There exist block diagonal $E$, $F$, $G$, and $H$, such that the linearized dynamics of (\ref{eq:RDfeedback}) at the equilibrium $(x^*,0,0)$ satisfy the conditions of Theorem~\ref{thm:ofStabilization} in Appendix~\ref{sec:ofStabilization}, and hence the equilibrium is locally exponentially stable. The second part of Theorem~\ref{Thm:RDStabilization} follows from the robustness result of Theorem~\ref{thm:robustStability}.

\section{Learning  mixed-strategy NE with heterogeneous higher-order uncoupled dynamics}\label{sec:heterogeneous}

The previous section focused on learnability of completely mixed-strategy NE using higher-order replicator dynamics.
In this section, we extract the essential features that enable this stabilization.

\subsection{General sufficient conditions for learnability}\label{sec:general}

Let $x^*$ be an isolated completely mixed-strategy NE. An essential step in establishing learnability using replicator dynamics was rewriting higher-order replicator dynamics in the form of \eqref{eq:closedloopHORD}, defining both the control, $u_i$, and the measurement, $y_i$. This step, along with the point $(x^*, 0,0)$ being an equilibrium of \eqref{eq:overallRD}, enabled us to be in the setting of Appendix~\ref{sec:ofStabilization}.

We will follow the same steps, replacing replicator dynamics of player $i$ with some arbitrary learning dynamics $f_i$. In particular, we rewrite \eqref{eq:closedloopHORD} as follows:
\begin{subequations}\label{eq:IOHeterogeneous}
\begin{align}
\dot{x}_{i}&=f_i(x_i,P_{i}(x_{-i})+N_iu_i)\\
\dot{v}_i&=N_i\tr P_{i}(x_{-i}) -v_i   \\
y_i&=N_i\tr P_{i}(x_{-i}) -v_i.
\end{align}
\end{subequations}

Gathering state variables, along with measurements and inputs, we can write the overall system as follows:
\begin{subequations}\label{eq:overallHeterogeneous}
\begin{align}
\Pmatrix{\dot{x}\\ \dot{v}} &= \Pmatrix{\mathcal{L}(x,v,u)\\ \mathcal{N}\tr P(x) - v}    \\
y &= \mathcal{N}\tr P(x) - v.
\end{align}
\end{subequations}
Here, $\mathcal{F}(x,v,u)$  is replaced by $\mathcal{L}(x,v,u)$ indicating the change in players' dynamics. In particular, the $i^{\mathrm{th}}$ block row of $\mathcal{L}$ is
$$\mathcal{L}_i(x,v,u) = f_i(x_i,P_{i}(x_{-i})+N_iu_i).$$
It is now evident that the point $(x^*,0,0)$ must be an equilibrium of \eqref{eq:overallHeterogeneous}.

The second step was finding the collective reduced-order linear dynamics around the equilibrium $(x^*,0,0)$. Following the same steps, the collective reduced-order linear dynamics of \eqref{eq:overallHeterogeneous} around $(x^*,0,0)$ take the form
\begin{subequations}\label{eq:DecentralizedFormulationHeterogeneous}
\begin{align}
\Pmatrix{\dot{w} \\ \dot{v}}
&= \Pmatrix{\mathcal{N}\tr \bar{J} \mathcal{N}&0\\\mathcal{M} & -I}\Pmatrix{w\\ v} + \Pmatrix{\bdiag{N_1\tr \bar{B}_1,...,N_n\tr \bar{B}_n}\\0}u\\
y &= \Pmatrix{\mathcal{M}&-I}\Pmatrix{w\\ v},
\end{align}
\end{subequations}
where instead of $\mathcal{N}\tr J \mathcal{N}$ and $$\Pmatrix{\bdiag{N_1\tr B_1,...,N_n\tr B_n}\\0},$$
 we now have $\mathcal{N}\tr \bar{J} \mathcal{N}$ and  
$$\Pmatrix{\bdiag{N_1\tr \bar{B}_1,...,N_n\tr \bar{B}_n}\\0}.$$ 
Note that $f_i$ is not a function of $v_i$, which results in the zero upper-right block in the local dynamics matrix.

Finally, we analyzed the rank of \eqref{eq:DecentralizedtTestMatrix} for all required partitions. In analyzing the rank of \eqref{eq:DecentralizedtTestMatrix}, we used that both $\mathcal{N}\tr J \mathcal{N}$ and $N_i\tr B_i$,  $i=1,\hdots,n,$ were non-singular. We are now ready to state the main result.

\begin{Theorem}\label{Thm:HeterogeneousStabilization}
Consider a finite game with an isolated completely mixed-strategy NE, denoted by $x^*$, that satisfies Assumption~\ref{Assumption:Regularity}. For any standard-order uncoupled learning dynamics that satisfy the following:
\begin{itemize}
\item The point $(x^*,0,0)$ is an equilibrium of \eqref{eq:overallHeterogeneous}
\item The matrices $\mathcal{N}\tr \bar{J} \mathcal{N}$ and $ N_i\tr \bar{B}_i$, $i=1,\hdots,n$, are non-singular
\end{itemize}
there exists a higher-order uncoupled version of the dynamics such that the equilibrium $(x^*,0,0)$ is locally exponentially stable.
\end{Theorem}

\subsection{Target gradient play dynamics}

We will now use the results derived so far to study learnability of $x^*$ using target gradient play dynamics.

Standard-order target gradient play dynamics take the form
$$\dot{x}_i=\Pi_{\Delta}\left[x_i+p_i\right]-x_i,$$
where $\Pi_\Delta[x] = \arg \min_{s\in \Delta}\norm{x - s}.$
When in feedback with a game, the dynamics of player $i$ take the form
$$\dot{x}_i=\Pi_{\Delta}\left[x_i+P_i(x_{-i})\right]-x_i,$$
where the $x_{-i}$ also evolve according to target gradient play dynamics.

Before we proceed, let us look at higher-order extensions of target gradient play dynamics that have been explored in the literature. Reference \cite{shamma2005dynamic} introduces anticipatory target gradient play dynamics and demonstrates how anticipation can enable convergence to mixed-strategy NE. Reference \cite{toonsi2023Higher} uses a special class of higher-order target gradient play dynamics to study learnability of mixed-strategy NE in polymatrix games.

Returning to the analysis, suppose that all players in a finite game utilize target gradient play dynamics. Following previous procedure to study learnability, we write:
\begin{subequations}\label{eq:targetgradient}
\begin{align}
\dot{x}_{i}&=\Pi_{\Delta}\left[x_i+P_i(x_{-i})+N_iu_i\right]-x_i\\
\dot{v}_i&=N_i\tr P_{i}(x_{-i}) -v_i   \\
y_i&=N_i\tr P_{i}(x_{-i}) -v_i.
\end{align}
\end{subequations}
The point $(x^*,0,0)$ is an equilibrium of the overall dynamics.
The collective reduced-order linear dynamics around $(x^*,0,0)$ take the form
\begin{align*}
\Pmatrix{\dot{w}\\\dot{v}}&=\Pmatrix{\mathcal{M} & 0 \\ \mathcal{M} & -I}\Pmatrix{w\\v}+  \Pmatrix{I\\ 0} u\\
y&= \Pmatrix{\mathcal{M}& -I}\Pmatrix{w\\v},
\end{align*}
where in this case
$$\mathcal{N}\tr \bar{J} \mathcal{N} =\mathcal{M}, \quad \bar{B}_i=N_i,$$
and $\mathcal{M}$ has the structure in \eqref{eq:calligraphicM}. With this, it is direct to see that target gradient play dynamics satisfy the requirements of Theorem~\ref{Thm:HeterogeneousStabilization}.

\subsection{Illustration: Replicator + target gradient play}

We now demonstrate learnability of $x^*$ using heterogeneous dynamics, where some players use replicator dynamics and others use target gradient play.

\begin{Proposition}
Any combination of replicator dynamics and target gradient play dynamics satisfies the assumptions of Theorem~\ref{Thm:HeterogeneousStabilization}.
\end{Proposition}

\begin{proof}
Without loss of generality, assume that the first $\ell$ players use replicator dynamics. Then, we have the following:
\begin{align*}
\mathcal{L}_i(x,v,u)&=\diag{P_{i}(x_{-i})+N_i u_i -\left( x_{i}\tr (P_{i}(x_{-i})+N_i u_i) \right)\one}x_{i}, \quad &\text{for }  i=1,\hdots, \ell \\
\mathcal{L}_i(x,v,u)&=\Pi_{\Delta}\left[x_i+P_i(x_{-i})+N_iu_i\right]-x_i, \quad &\text{for }  i= \ell+1,\hdots, n.
\end{align*}

Following previous process, we have
$$\bar{J}_{ij} = \diag{x_i^*}\left(M_{ij} - \one (x_i^*)\tr M_{ij}\right) \quad \text{and} \quad \bar{B}_i= \diag{x_i^*}(N_i - \one (x_i^*)\tr N_i), \quad \text{for }  i=1,\hdots, \ell.$$

For players playing target gradient play, we have
$$\bar{J}_{ij}=M_{ij} \quad \text{and} \quad \bar{B}_i=N_i, \quad \text{for }  i=\ell+1,\hdots,n.$$

The collective local linear dynamics take the form:
\begin{align*}
\Pmatrix{\dot{w} \\ \dot{v}}
&= \Pmatrix{\mathcal{N}\tr \bar{J} \mathcal{N}&0\\\mathcal{M} & -I}\Pmatrix{w\\ v} + \Pmatrix{\bdiag{N_1\tr \bar{B}_1,\hdots,N_\ell \tr \bar{B}_\ell, I,\hdots, I}\\0}u\\
y &= \Pmatrix{\mathcal{M}&-I}\Pmatrix{w\\ v}.
\end{align*}

We now show that $\mathcal{N}\tr \bar{J} \mathcal{N}$ is non-singular when Assumption~\ref{Assumption:Regularity} holds. Suppose
$$\mathcal{N}\tr \bar{J} \mathcal{N} c=0$$
for a non-zero, suitably partitioned, vector $c = (c_1,...,c_n)$. We directly have
$$\mathcal{M}_i c=0\quad \text{for }  i=\ell+1,\hdots,n.$$

Following the same arguments used to prove Claim~\ref{Claim:nonSingular}, it must be that
$$\sum_{j\not= i} M_{ij}N_j c_j = \sum_{j\not= i} \mathbf{1}(x_i^*)\tr M_{ij}N_j c_j, \quad \text{for }  i=1,\hdots,\ell.$$
Multiplying on the left by $N_i\tr$ results in
$$\sum_{j\not=i} N_i\tr M_{ij} N_j c_j = \sum_{j\not= i} N_i\tr \one (x_i^*)\tr M_{ij}N_j c_j = 0,  \quad \text{for }  i=1,\hdots,\ell$$
and therefore
$$\mathcal{M}_ic=0  \quad \text{for }  i=1,\hdots,\ell.$$

Collecting everything, we see that for the same $c$
$$\mathcal{N}\tr \bar{J} \mathcal{N} c=0 \implies  \mathcal{M} c=0.$$

With this non-singularity proof, we conclude learnability of $x^*$ using such heterogeneous dynamics.
\end{proof}

\section{Simultaneous stabilization of two mixed-strategy NE} \label{Sec:Simultaneous}\ \\[-15pt]

In this section, we address a learning limitation that persists for higher-order constructions. In particular, we show how higher-order learning dynamics cannot simultaneously learn the completely mixed-strategy NE of two given games. Here, ``simultaneously'' means that if learning dynamics are stable for one, then they cannot be stable for the other. These results hold for general higher-order learning dynamics that satisfy minimal assumptions.

\subsection{Two games with different network structures}

Let us begin by discussing the class of games of interest.
To this end, define $\Gamma_\mathrm{CY}(c)$ to be the following (cyclic) polymatrix game
$$R_1(x_1,x_2) = x_1\tr M_1(c_1)  x_2\quad
R_2(x_2,x_3) = x_2\tr M_2(c_2) x_3$$
$$R_3(x_3,x_4) = x_3\tr M_3(c_3) x_4\quad
R_4(x_4,x_1) = x_4\tr M_4(c_4) x_1,$$
and define $\Gamma_\mathrm{PW}(c)$ to be the following (pairwise) polymatrix game
$$R_1(x_1,x_4) = x_1\tr M_1(c_1) x_4\quad
R_2(x_2,x_3) = x_2\tr M_2(c_2) x_3 $$
$$ R_3(x_3,x_2) = x_3\tr M_3(c_3) x_2 \quad
R_4(x_4,x_1) = x_4\tr M_4(c_4) x_1.$$
The utility matrices of all players have the same dimension $\mathbb{R}^{\ell \times \ell}$, $c=(c_1,c_2,c_3,c_4)$, and each $c_i$ is a tuple of $\ell-1$ scalar elements, such that $c_i=(c_{i1},\hdots,c_{i(\ell-1)})$.
Each utility matrix, $M_i$, is a function of the utility parameters, $c_i$.
For simplicity, assume each element of $c_i$ affects a distinct entry of $M_i$ through a non-trivial linear function.
The utility matrices of all players are the same in both games. Only the network structures are different (see Figure~\ref{fig:Twogameswithdiffstruc}).
We will not define the players' exact utility matrices. Instead, we will consider a list of requirements that both games must satisfy.

\begin{figure}[!t]
\centerline{\includegraphics[width=.7\columnwidth]{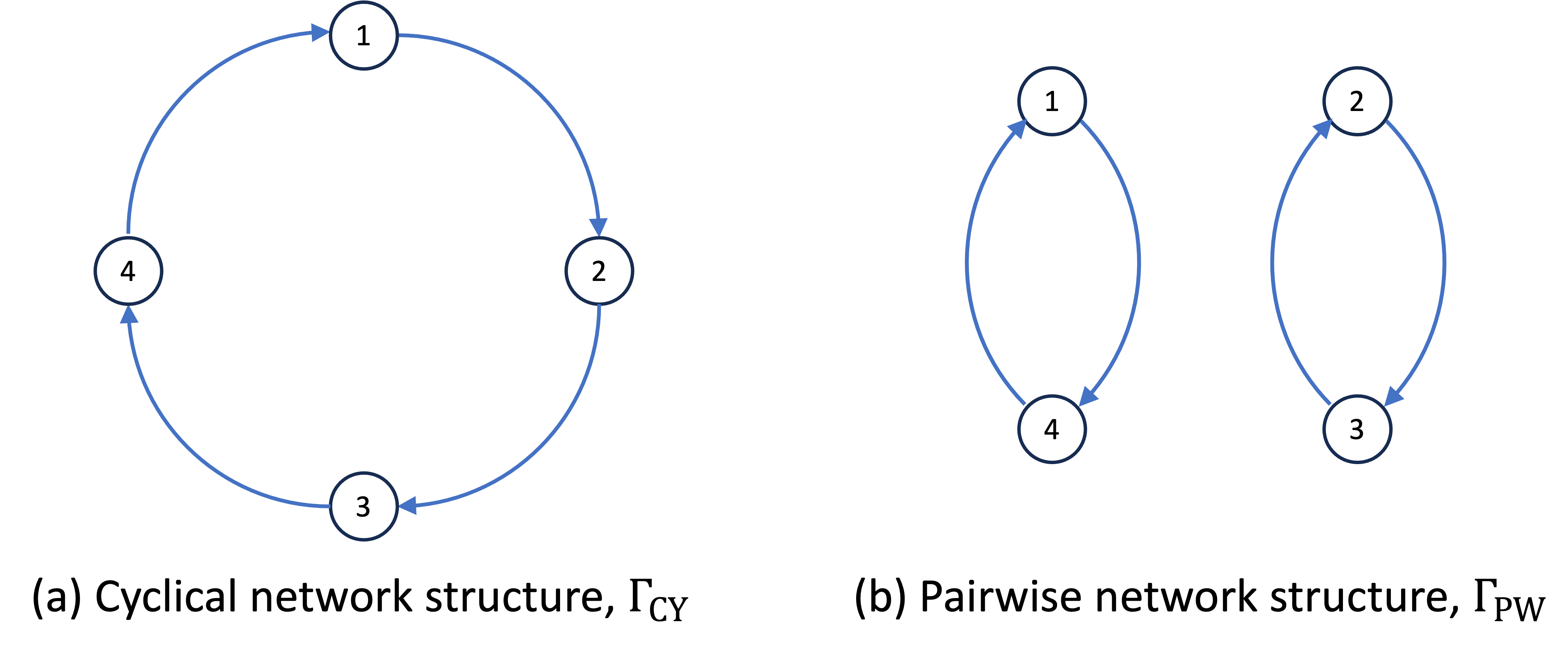}}
\caption{Network structures of $\Gamma_\mathrm{CY}$ and $\Gamma_\mathrm{PW}$. Arrows indicate strategic dependency (e.g., in $\Gamma_\mathrm{CY}$, player 1's payoff depends on player 2's strategy.)}
\label{fig:Twogameswithdiffstruc}
\end{figure}

Let $c^*=(c^*_1,c^*_2,c^*_3,c^*_4)$ denote the nominal values of the utility parameters. Define the set
$$\mathbf{c}_\epsilon = \theset{(c_1,c_2,c_3,c_4) \st \magn{c_{ik} - c^*_{ik}} < \epsilon},$$
which generates a neighborhood of games around any of the two nominal games $\Gamma_\mathrm{CY}(c^*)$ or $\Gamma_\mathrm{PW}(c^*)$. If the nominal game has an isolated completely mixed-strategy NE, then for a sufficiently small $\epsilon$, each game in the neighborhood has an isolated completely mixed-strategy NE parameterized by the tuple $c$.

Let us take a closer look at the equilibria of the games as functions of the utility parameters.
Suppose $x^*(c)$ is an isolated completely mixed-strategy NE of $\Gamma_\mathrm{CY}(c)$, then we have
$$x_1^*(c_4)=\Pmatrix{x_{11}^*(c_4)\\ \vdots \\ x_{1\ell}^*(c_4)}\quad
x_2^*(c_1)=\Pmatrix{x_{21}^*(c_1)\\ \vdots \\ x_{2\ell}^*(c_1)}\quad
x_3^*(c_2)=\Pmatrix{x_{31}^*(c_2)\\ \vdots \\ x_{3\ell}^*(c_2)} \quad
x_4^*(c_3)=\Pmatrix{x_{41}^*(c_3)\\ \vdots \\ x_{4\ell}^*(c_3)} ,$$
where  $x^*_{ik} (c_j)$ denotes the $k^\text{th}$ entry of the NE strategy of player $i$ in $\Gamma_\mathrm{CY}(c)$, and $j$ is the index of the player whose utility is affected by the strategy of player $i$.
Similarly, suppose $y^*(c)$  is an isolated completely mixed-strategy NE  of $\Gamma_\mathrm{PW}(c)$, then we have
$$y_1^*(c_4)=\Pmatrix{y_{11}^*(c_4)\\ \vdots \\ y_{1\ell}^*(c_4)} \quad
y_2^*(c_3)=\Pmatrix{y_{21}^*(c_3)\\ \vdots \\ y_{2\ell}^*(c_3)}\quad
y_3^*(c_2)=\Pmatrix{y_{31}^*(c_2)\\ \vdots \\ y_{3\ell}^*(c_2)} \quad
y_4^*(c_1)=\Pmatrix{y_{41}^*(c_1)\\ \vdots \\ y_{4\ell}^*(c_1)} .$$
Note the dependence on opponents' utility parameters given the game structure.

Let us now list the requirements for both games.

\begin{Assumption}\label{Assumption:Requirements} \ \\[-20pt]
\begin{itemize}
\item[a.] The utility matrices of all players have the same dimension $\mathbb{R}^{\ell \times \ell}$, where $\ell$ is an even integer.
\item[b.] Both $\Gamma_\mathrm{CY}(c^*)$ and $\Gamma_\mathrm{PW}(c^*)$ have an isolated completely mixed-strategy NE  (either game could have other types of NE)
denoted by $x^*(c^*)$ and $y^*(c^*)$ respectively .
\item[c.]  Let $x^*_i(c_j)$ denote the NE strategy of player $i$ in $\Gamma_\mathrm{CY}(c)$,  where $j$ is the index of the player whose utility is affected by the strategy of player $i$, and let $N$ be defined as in \eqref{eq:Ncondition}. We require the matrix
$$\Pmatrix{N\tr \nabla_{c_{j1}} x_i^*(c^*_j)&\hdots& N\tr \nabla_{c_{j(\ell-1)}} x_i^*(c^*_j)}$$
to have rank $\ell-1$ for each player $i$.
\item[d.] $x^*_2(c^*_1)=x^*_4(c^*_3)$
\end{itemize}
\end{Assumption}
Since the utility matrices of the players are the same in $\Gamma_\mathrm{CY}(c^*)$ and $\Gamma_\mathrm{PW}(c^*)$, Assumption~\ref{Assumption:Requirements}-d implies that $x^*(c^*)=y^*(c^*)$. In particular, we have

$$
x_1^*(c_4^*)=y_1^*(c_4^*), \quad
x_3^*(c_2^*)=y_3^*(c_2^*),\quad \text{and} \quad  x_2^*(c_1^*)=x_4^*(c_3^*)=y_2^*(c_3^*)=y_4^*(c_1^*).
$$

\begin{Example}
Let the utility matrices be
$$M_1(c_1)=\Pmatrix{c_{11} & 0 & 0&0\\
0 & c_{12} & 0&0\\
0 & 0 & c_{13} &0\\
0 & 0 & 0&1},
M_2(c_2)=\Pmatrix{c_{21} & 0 & 0&0
\\0 & c_{22} & 0&0
\\0 & 0 & c_{23}&0
\\ 0 & 0 & 0&1},$$
$$M_3(c_3)=\Pmatrix{0 & 0 & c_{31}&0\\
0 & 0 & 0&c_{32}\\
c_{33} & 0 & 0&0\\
0 & 1 & 0&0},
M_4(c_4)=\Pmatrix{0 & c_{41} & 0&0\\
0 & 0 & c_{42}&0\\
0 & 0 & 0& c_{43}\\
1 & 0 & 0&0}.$$
The nominal values of the parameters are $c^*_1=c^*_2=c^*_3=c^*_4=(1,1,1)$. The nominal games $\Gamma_\mathrm{CY}(c^*)$ and $\Gamma_\mathrm{PW}(c^*)$ have an isolated NE at
$$x_i^*(c^*)=\Pmatrix{\frac{1}{4}&\frac{1}{4}&\frac{1}{4}&\frac{1}{4}}\tr \quad \text{for } i=1,2,3,4.$$
The game $\Gamma_\mathrm{PW}(c^*)$ has two other NE that are not completely mixed.
The local dynamics matrix of replicator dynamics around $x^*$ in $\Gamma_\mathrm{CY}(c^*)$ has no real eigenvalues. An anticipatory version of replicator dynamics can lead to $x^*$ in $\Gamma_\mathrm{CY}(c^*)$.
The same dynamics cannot lead to $x^*$ in $\Gamma_\mathrm{PW}(c^*)$. In fact, the local dynamics matrix of replicator dynamics in $\Gamma_\mathrm{PW}(c^*)$ around $x^*$ has a real positive eigenvalue.
\end{Example}

\subsection{Impossibility of simultaneous stabilization}

Next, we will discuss the assumptions on the learning dynamics of the players. Let the learning dynamics of the four players be
\begin{align*}
\dot{x}_1 &= f_1(x_1,p_1,\phi_1(p_1,z_1)) &
\dot{x}_2 &= f_2(x_2,p_2,\phi_2(p_2,z_2)) & \\
\dot{z}_1 &= g_1(p_1,z_1) &
\dot{z}_2 &= g_2(p_2,z_2)
\end{align*}
\begin{align*}
\dot{x}_3 &= f_3(x_3,p_3,\phi_3(p_3,z_3)) &
\dot{x}_4 &= f_4(x_4,p_4,\phi_4(p_4,z_4)) & \\
\dot{z}_3 &= g_3(p_3,z_3) &
\dot{z}_4 &= g_4(p_4,z_4) .
\end{align*}
Here, the effect of $\phi_i$ on the function $f_i$ does not have to be additive.

When the players are playing $\Gamma_\mathrm{CY}$, the input vector $p_i(t)$ of each player becomes
\begin{align*}
p_1(t)&=P_1(x_2(t))=M_1(c_1) x_2(t), \quad p_2(t)=P_2(x_3(t))=M_2(c_2)x_3(t), \\ p_3(t)&=P_3(x_4(t))=M_3(c_3)x_4(t), \quad p_4(t)=P_4(x_1(t))=M_4(c_4) x_1(t),
\end{align*}
while the input vectors in $\Gamma_\mathrm{PW}$  are
\begin{align*}
p_1(t)&=P_1(x_4(t))=M_1(c_1)x_4(t), \quad p_2(t)=P_2(x_3(t))=M_2(c_2) x_3(t), \\
p_3(t)&=P_3(x_2(t))=M_3(c_3) x_2(t), \quad p_4(t)=P_4(x_1(t))=M_4(c_4)x_1(t).
\end{align*}
Moving forward, we will omit the explicit dependence of $M_i(c_i)$ and simply write $M_i$ for brevity, with the understanding that it is still a function of $c_i$.

We will now look at a reduced-order version of the dynamics for both games. Let us start with $\Gamma_\mathrm{CY}(c)$. Recalling equation~\eqref{eq:nearby}, we define
$$x_i(t) =x_i^*(c^*_j)+ N w_i(t)$$
where $w_i(t)$ is unique and satisfies
$${w}_i(t) = N\tr {x}_i(t)-N\tr x_i^*(c^*_j),$$
and $N$ is defined according to \eqref{eq:Ncondition}.
With this definition, let
\begin{subequations}\label{eq:red}
\begin{align}
\dot{w} = f^\mathrm{red}_\mathrm{CY}(w,z;c)\\
\dot{z} = g^\mathrm{red}_\mathrm{CY}(w,z;c)
\end{align}
\end{subequations}
denote the reduced-order dynamics of the game $\Gamma_\mathrm{CY}(c)$. These reduced-order dynamics represent the deviation of the strategies in $\Gamma_\mathrm{CY}(c)$ from the NE of the nominal game, where $c\in \mathbf{c}_\epsilon$ and $\epsilon$ is sufficiently small.

Following the same procedure, define
$$x_i(t) =x_i^*(c^*_j)+ N q_i(t)$$ for players' strategies in  $\Gamma_\mathrm{PW}(c)$,
where $q_i(t)$ is unique and satisfies
$${q}_i(t)= N\tr {x}_i(t)-N\tr x_i^*(c^*_j).$$
Let the dynamics
\begin{align*}
\dot{q} = f^\mathrm{red}_\mathrm{PW}(q,z;c)\\
\dot{z} = g^\mathrm{red}_\mathrm{PW}(q,z;c),
\end{align*}
denote the reduced-order dynamics of the game $\Gamma_\mathrm{PW}(c)$, where $c\in \mathbf{c}_\epsilon$ and $\epsilon$ is sufficiently small. Again, these dynamics represent the deviation of the strategies in $\Gamma_\mathrm{PW}(c)$ from the NE of the nominal game.

Now define
$$w_1^*(c_4)=N\tr \left(x_1^*(c_4)- x_1^*(c^*_4) \right), \quad
w_2^*(c_1)=N\tr \left( x_2^*(c_1)- x_2^*(c^*_1)\right),$$
$$ w_3^*(c_2)=N\tr \left( x_3^*(c_2)-  x_3^*(c^*_2)\right),  \quad
w_4^*(c_3)=N\tr \left( x_4^*(c_3)- x_4^*(c^*_3)\right)$$
$$q_1^*(c_4)=N\tr \left(y_1^*(c_4)- x_1^*(c^*_4) \right), \quad
q_2^*(c_3)=N\tr \left( y_2^*(c_3)- x_2^*(c^*_1)\right), $$
$$ q_3^*(c_2)=N\tr \left( y_3^*(c_2)-  x_3^*(c^*_2)\right),  \quad
q_4^*(c_1)=N\tr \left( y_4^*(c_1)- x_4^*(c^*_3)\right). $$
The tuple $w^*(c)=(w_1^*(c_4), w_2^*(c_1),w_3^*(c_2), w_4^*(c_3))$ represents the deviation between the isolated mixed-strategy NE of the nominal game, $\Gamma_\mathrm{CY}(c^*)$, and the isolated mixed-strategy NE of the perturbed game, $\Gamma_\mathrm{CY}(c)$, while $q^*(c)=(q_1^*(c_4), q_2^*(c_3),q_3^*(c_2), q_4^*(c_1))$ represents the deviation between the isolated mixed-strategy NE of the nominal game, $\Gamma_\mathrm{PW}(c^*)$, and the isolated mixed-strategy NE of the perturbed $\Gamma_\mathrm{PW}(c)$.

\begin{Assumption}\label{Assumption:NEStationarity}
For sufficiently small $\epsilon > 0$, the dynamics satisfy the following:
\begin{itemize}
\item For each $c\in \mathbf{c}_\epsilon$,  there exist suitably defined and unique
$$r^*(c)=(r_1^*(c_1), r_2^*(c_2), r_3^*(c_3), r_4^*(c_4)), \quad \text{and} \quad s^*(c)=(s_1^*(c_1), s_2^*(c_2), s_3^*(c_3), s_4^*(c_4)),$$
such that
\begin{align*}
0 &= f^\mathrm{red}_\mathrm{CY}(w^*(c),r^*(c);c),&\quad 0& = g^\mathrm{red}_\mathrm{CY}(w^*(c),r^*(c);c) \\
0&= f^\mathrm{red}_\mathrm{PW}(q^*(c),s^*(c);c),&\quad 0&= g^\mathrm{red}_\mathrm{PW}(q^*(c),s^*(c);c).
\end{align*}
\item The reduced-order dynamics are continuously differentiable in all their arguments within a region near the equilibria $(w^*(c),r^*(c))$ and $(q^*(c),s^*(c))$ for each $c\in \mathbf{c}_\epsilon$.
\end{itemize}
\end{Assumption}

In words, Assumption~\ref{Assumption:NEStationarity} assumes that the dynamics, in all sufficiently small perturbed games, are stationary at the completely mixed-strategy NE, with suitably defined higher-order equilibrium components.
Also, recall that $x^*_2(c^*_1)=x^*_4(c^*_3)$ and $w^*(c^*)=q^*(c^*)$.
Therefore, we have $r^*(c^*)=s^*(c^*)$ by the uniqueness assumption on the higher-order equilibrium components.
We are now ready to state the main result.

\begin{Theorem}\label{Thm:SStabilization}
Under Assumptions~\ref{Assumption:Requirements} and \ref{Assumption:NEStationarity}, learning dynamics cannot be locally exponentially stable at $(w^*(c^*), r^*(c^*))$ in $\Gamma_\mathrm{CY}(c^*)$ and at $(w^*(c^*), r^*(c^*))$ in $\Gamma_\mathrm{PW}(c^*)$ simultaneously.
\end{Theorem}

The remainder of this section is devoted to the proof of Theorem~\ref{Thm:SStabilization}.

\subsection{The structure of the Jacobian matrices}

As was done previously, we will use linear system techniques to study the local behavior of nonlinear systems. In particular, we will study local stability through linearization around equilibrium and check whether the local linear dynamics matrix is a stability matrix. Let us now study the structures of the matrices of the local linear dynamics.

To this end, let $J_{\mathrm{CY}}(w,z)$ denote the Jacobian matrix of $\Pmatrix{f^\mathrm{red}_\mathrm{CY}\\ g^\mathrm{red}_\mathrm{CY}}$ evaluated at $(w,z)$ in the nominal game, and $J_{\mathrm{PW}}(q,z)$ denote the Jacobian matrix of $\Pmatrix{f^\mathrm{red}_\mathrm{PW}\\ g^\mathrm{red}_\mathrm{PW}}$ evaluated at $(q,z)$ in the nominal game.
\begin{Claim}\label{Claim:JacobianStructureGamma}
The Jacobian matrix $J_{\mathrm{CY}}(0,r^*)$ has the structure
$$J_{\mathrm{CY}}(0,r^*)= \Pmatrix{
0 & J_{1} & 0 & 0 & G_1 & 0 & 0 & 0\\
0 & 0 & J_{2} & 0 & 0 & G_2 & 0 & 0\\
0 & 0 & 0 & J_{3} & 0 & 0 & G_3 & 0\\
J_{4} & 0 & 0 & 0 & 0 & 0 & 0 & G_4\\
0 & F_{1} & 0 & 0 & E_1 & 0 & 0 & 0\\
0 & 0 & F_{2} & 0 & 0 & E_2 & 0 & 0\\
0 & 0 & 0 & F_{3} & 0 & 0 & E_3 & 0\\
F_{4} & 0 & 0 & 0 & 0 & 0 & 0 & E_4\\},$$
where
\begin{alignat*}{2}
J_{1}&=\nabla_{w_2}f_{\mathrm{CY},1}^\mathrm{red} (0,r^*;c^*),\quad &
J_{2}&=   \nabla_{w_3}f_{\mathrm{CY},2}^\mathrm{red} (0,r^*;c^*),\\
J_{3}&= \nabla_{w_4}f_{\mathrm{CY},3}^\mathrm{red} (0,r^*;c^*),\quad&
J_{4}&=\nabla_{w_1}f_{\mathrm{CY},4}^\mathrm{red} (0,r^*;c^*) ,\\
F_{1}&=\nabla_{w_2}g_{\mathrm{CY},1}^\mathrm{red} (0,r^*;c^*),\quad&
F_{2}&=   \nabla_{w_3}g_{\mathrm{CY},2}^\mathrm{red} (0,r^*;c^*),\\
F_{3}&= \nabla_{w_4}g_{\mathrm{CY},3}^\mathrm{red} (0,r^*;c^*),\quad&
F_{4}&=\nabla_{w_1}g_{\mathrm{CY},4}^\mathrm{red} (0,r^*;c^*) ,\\
G_{i}&=\nabla_{z_i}f_{\mathrm{CY},i}^\mathrm{red} (0,r^*;c^*), \quad& E_{i}&=\nabla_{z_i}g_{\mathrm{CY},i}^\mathrm{red} (0,r^*;c^*).
\end{alignat*}

\end{Claim}

\begin{proof}
Most of the zero blocks are direct consequences of the lack of explicit dependence. The only exceptions are $\nabla_{w_i}f_{\mathrm{CY},i}^\mathrm{red} (0,r^*;c^*)$. Let us first prove the proposition for the case when $\ell=2$.

We will consider $\nabla_{w_1}f_{\mathrm{CY},1}^\mathrm{red} $. The others will follow from similar arguments. The NE, being stationary points of the dynamics, imply that
\begin{align*}
&f_{\mathrm{CY},1}^\mathrm{red}(w^*(c),r^*(c);c) \\&\quad= N\tr f_{\mathrm{CY},1}\left(x_1^*(c^*_4) + Nw_1^*(c_4), M_{1}\cdot \left(x_2^*(c^*_1)+ Nw_2^*(c_1)\right), \phi_1\left( M_{1}\cdot (x_2^*(c^*_1)+ Nw_2^*(c_1)),r_1^*(c_1)\right);c\right)\\&\quad= 0,
\end{align*}
for all $c\in\mathbf{c}_\epsilon$, where $\epsilon$ is sufficiently small.

Taking the total derivative with respect to $c_4$ results in
$$\frac{\partial f_{\mathrm{CY},1}^\mathrm{red}}{\partial w_1} \frac{\partial w_1^*}{\partial c_4}\equiv 0,$$
which holds for all $c\in \mathbf{c}_\epsilon$.
In particular, we have
$$\nabla_{w_1}f_{\mathrm{CY},1}^\mathrm{red}(0,r^*(c^*);c^*) N\tr \nabla_{c_4}x_1^*(c_4^*)=0.$$
By Assumption~\ref{Assumption:Requirements}-c, it must be that
$$\nabla_{w_1}f_{\mathrm{CY},1}^\mathrm{red}(0,r^*(c^*);c^*) = 0.$$

Similar arguments show that
$$\nabla_{w_2}f_{\mathrm{CY},2}^\mathrm{red}(0,r^*(c^*);c^*)= 0,\quad
\nabla_{w_3}f_{\mathrm{CY},3}^\mathrm{red}(0,r^*(c^*);c^*) = 0,\quad
\nabla_{w_4}f_{\mathrm{CY},4}^\mathrm{red}(0,r^*(c^*);c^*) = 0.$$
Note that taking the total derivative with respect to $c_1$ does not imply that $\frac{\partial f_{\mathrm{CY},1}^\mathrm{red}}{\partial w_2} = 0$, since $M_{1}$ depends on $c_1$, which is not shown explicitly.

Let us now consider the case where $\ell>2$, using arguments similar to those in the scalar case. We will first take the total derivative with respect to $c_{41}$, which gives
$$\nabla_{w_1}f_{\mathrm{CY},1}^\mathrm{red}(0,r^*(c^*);c^*) N\tr \nabla_{c_{41}}x^*_1(c_4^*)=0.$$
Taking the total derivative with respect to $c_{42},....,c_{4(\ell-1)}$ results in
\begin{align*}
\nabla_{w_1}f_{\mathrm{CY},1}^\mathrm{red}(0,r^*(c^*);c^*) N\tr \nabla_{c_{42}}x^*_1(c_4^*)&= 0\\
&\vdots \\
\nabla_{w_1}f_{\mathrm{CY},1}^\mathrm{red}(0,r^*(c^*);c^*) N\tr \nabla_{c_{4(\ell-1)}}x^*_1(c_4^*)&= 0.
\end{align*}
Again,  by Assumption~\ref{Assumption:Requirements}-c, it  must be that
$$\nabla_{w_1}f_{\mathrm{CY},1}^\mathrm{red}(0,r^*(c^*);c^*)=0.$$
\end{proof}

Now let us inspect the structure of $J_\mathrm{PW}(0,r^*)$.

\begin{Claim}\label{Claim:JacobianStructureOmega}
The Jacobian matrix $J_{\mathrm{PW}}(0,r^*)$ has the structure
$$J_{\mathrm{PW}}(0,r^*)= \Pmatrix{
0 & 0 & 0 & J_{1} & G_1 & 0 & 0 & 0\\
0 & 0 & J_{2} & 0 & 0 & G_2 & 0 & 0\\
0 & J_{3} & 0 & 0 & 0 & 0 & G_3 & 0\\
J_{4} & 0 & 0 & 0 & 0 & 0 & 0 & G_4\\
0 & 0 & 0 & F_{1} & E_1 & 0 & 0 & 0\\
0 & 0 & F_{2} & 0 & 0 & E_2 & 0 & 0\\
0 & F_{3} & 0 & 0 & 0 & 0 & E_3 & 0\\
F_{4} & 0 & 0 & 0 & 0 & 0 & 0 & E_4}, $$
with the same block matrices in $J_{\mathrm{CY}}(0,r^*).$
\end{Claim}

\begin{proof}
Using similar arguments to the proof of Claim~\ref{Claim:JacobianStructureGamma}, we conclude that  $\nabla_{q_i}f_{\mathrm{PW},i}^\mathrm{red} (0,r^*)=0$. Let us now inspect the relation between the block matrices of $J_{\mathrm{CY}}(0,r^*)$ and $J_{\mathrm{PW}}(0,r^*)$.

We will first look at the local dynamics of player 1 in $\Gamma_\mathrm{CY}(c^*)$:
\begin{align*}
\dot{w}_1&=f_{\mathrm{CY},1}^\mathrm{red}(w,z;c^*)\\
\dot{z}_1&=g_{\mathrm{CY},1}^\mathrm{red}(w,z;c^*)
\end{align*}
which are
\begin{align*}
\dot{w}_{1}&=N\tr f_1\left(x_1^*(c^*_4) + Nw_1, M_{1}\cdot \left(x_2^*(c^*_1)+ Nw_2\right),\phi_1 \left( M_{1}\cdot(x_2^*(c^*_1)+ Nw_2),z_1\right)\right)\\
\dot{z}_1&=g_{\mathrm{CY},1}^\mathrm{red}(M_{1}\cdot \left(x_2^*(c^*_1)+ Nw_2\right),z_1).
\end{align*}

Similarly, the local dynamics of player 1 in $\Gamma_\mathrm{PW}(c^*)$ take the form
\begin{align*}
\dot{q}_1&=f_{\mathrm{PW},1}^\mathrm{red}(q,z;c^*)\\
\dot{z}_1&=g_{\mathrm{PW},1}^\mathrm{red}(q,z;c^*)
\end{align*}
where
\begin{align*}
\dot{q}_1&=N\tr f_1\left(y_1^*(c^*_4) + Nq_1, M_{1}\cdot \left(y_4^*(c^*_1)+ Nq_4\right),\phi_1 \left( M_{1} \cdot  \left(y_4^*(c^*_1)+ Nq_4\right),z_1\right)\right)\\
\dot{z}_1&=g_{\mathrm{PW},1}^\mathrm{red}(M_{1} \cdot \left(y_4^*(c^*_1)+ Nq_4\right),z_1).
\end{align*}
It is clear that
$$
\nabla_{w_2} f_{\mathrm{CY},1}^\mathrm{red} (w_1,w_2,z_1)= \nabla_{q_4} f_{\mathrm{PW},1}^\mathrm{red} (q_1,q_4,\bar{z}_1), \quad \text{and} \quad \nabla_{w_2} g_{\mathrm{CY},1}^\mathrm{red} (w_2,z_1)= \nabla_{q_4} g_{\mathrm{PW},1}^\mathrm{red} (q_4,\bar{z}_1),
$$
for $w_1=q_1$, $w_2=q_4$, and $z_1=\bar{z}_1$, which is satisfied at $(0,r^*)$.
Follow the same procedure for player 3.

Next, inspect the local dynamics of player 2 in $\Gamma_\mathrm{CY}(c^*)$:
\begin{align*}
\dot{w}_{2}&=N\tr f_2\left(x_2^*(c^*_1) + Nw_2, M_{2}\cdot \left(x_3^*(c^*_2)+ Nw_3\right),\phi_1 \left( M_{2}\cdot(x_3^*(c^*_2)+ Nw_3),z_2\right)\right)\\
\dot{z}_2&=g_{\mathrm{CY},2}^\mathrm{red}(M_{2}\cdot \left(x_3^*(c^*_2)+ Nw_3\right),z_2),
\end{align*}
and the local dynamics of player 2 in $\Gamma_\mathrm{PW}(c^*)$:
\begin{align*}
\dot{q}_2&=N\tr f_2\left(y_2^*(c^*_3) + Nq_2, M_{2}\cdot \left(y_3^*(c^*_2)+ Nq_3\right),\phi_1 \left( M_{2} \cdot  \left(y_3^*(c^*_2)+ Nq_3\right),z_2\right)\right)\\
\dot{z}_2&=g_{\mathrm{PW},2}^\mathrm{red}(M_{2} \cdot \left(y_3^*(c^*_2)+ Nq_3\right),z_2).
\end{align*}
Recall that $x_2^*(c_1^*)=y_2^*(c_3^*)$,  which completes the proof for player 2. Follow the same procedure for player 4.
\end{proof}

\subsection{Characteristic polynomials and constant terms}

To find eigenvalues of each Jacobian we need to compute
$$\psi_{\mathrm{CY}}(\lambda)=\Det{J_{\mathrm{CY}} (0,r^*)-\lambda I }=0,\quad \text{and} \quad \psi_{\mathrm{PW}}(\lambda)=\Det{J_{\mathrm{PW}} (0,r^*)-\lambda I}=0. $$

We now show that the (monic) polynomials, $\psi_{\mathrm{CY}}(\lambda)$ and $\psi_{\mathrm{PW}}(\lambda)$, are such that the constant terms $\psi_{\mathrm{CY}}(0)$ and $\psi_{\mathrm{PW}}(0)$ have opposite signs. Accordingly, it cannot be that both have roots with negative real parts (since any such polynomial must have strictly positive coefficients, e.g., \cite[Theorem~9.1]{poznyak2008stable}).

The constant terms take the form
$$\psi_{\mathrm{CY}}(0)=\Det{ J_{\mathrm{CY}} (0,r^*)}\quad \text{and} \quad \psi_{\mathrm{PW}}(0)=\Det{ J_{\mathrm{PW}} (0,r^*)}. $$
The matrix  $J_{\mathrm{CY}} (0,r^*)$ is $\ell-1$ column swaps away from $J_{\mathrm{PW}} (0,r^*)$. Since $\ell$ is even, these must have opposite signs. This concludes the proof of Theorem~\ref{Thm:SStabilization}.

\section{The asymptotic best-response property}\label{Sec:ABR}\ \\[-15pt]

Up to this point, the discussion has been on proving learnability. No restrictions were set on the dynamics to ensure reasonable behavior. In this section, we introduce Asymptotic Best-Response (ABR), a property inherent to reasonable learning dynamics.

\subsection{ABR and internal stability}

Let us define the Asymptotic Best-Response (ABR) property for the standard-order dynamics defined in \eqref{eq:standardorder}. This property states that if
$$p_i(t) \rightarrow \bar{p},$$
as $t\rightarrow \infty$, where $\bar{p}$ is a constant vector,
then
$$x_i(t)\tr p_i(t) \rightarrow \max_{x_i\in \Delta(k_i)} x_i \tr \bar{p}.$$
In words, the ABR property stipulates that dynamics converge to a best response of an asymptotically stationary environment. If $\bar{p}$ has a unique maximal element, then $x_i(t)$ must converge to the corresponding vertex. If $\bar{p}$ is a vector with all equal values, e.g., $\bar{p} = \alpha \mathbf{1}$, then there is no asymptotic constraint on $x_i(t)$. In particular, $x_i(t)$ itself need not converge.

For higher-order learning dynamics, we will establish a connection between ABR and the internal stability of the higher-order components.

\begin{Proposition}\label{Prop:Internallystable}
Consider the following higher-order learning dynamics
\begin{align*}
\dot{x}_i &= f_i(x_i, p_i+ \phi_i(p_i, z_i))\\
\dot{z}_i &=E_i z_i + F_i p_i\\
& \phi_i(p_i,z_i)=G_i z_i+H_i p_i,
\end{align*}
where Assumption~\ref{Assumption:higherorder} holds and $E_i$ is a stability matrix. Suppose that the standard-order dynamics
$$\dot{x}_i=f_i(x_i,p_i)$$
satisfy ABR. For the higher-order dynamics, if $p_i(t)\rightarrow \bar{p}$ as $t\rightarrow \infty$, then
$$x_i(t)\tr p_i(t) \rightarrow \max_{x_i\in \Delta(k_i)} x_i \tr \bar{p}.$$
\end{Proposition}

\begin{proof}
Let $\bar{p}$ be an arbitrary constant vector and suppose
$$p_i(t) \rightarrow \bar{p}.$$
Knowing that $E_i$ is a stability matrix, we have
$$z_i(t) \rightarrow - E_i^{-1} F_i \bar{p},\quad \text{and} \quad \phi_i(t) \rightarrow (-G_i E_i^{-1} F_i+H_i)\bar{p}.$$
Since Assumption~\ref{Assumption:higherorder} holds, it must be that
$$(-G_i E_i^{-1} F_i+H_i)\bar{p}=0.$$
In fact, we must have
$$-G_i E_i^{-1} F_i+H_i=0$$
since $\bar{p}$ is arbitrary.
Therefore
$$p_i+\phi_i \rightarrow \bar{p}.$$
The desired conclusion follows from the standard-order dynamics having the ABR property.
\end{proof}

Let us now inspect the issue of learning with unstable higher-order components. Consider the higher-order version of replicator dynamics presented in \eqref{eq:higherorderRD} for a specific set of matrices $(E_i,F_i,G_i,H_i)$. If $p_i(t) \equiv \bar{p}$ and $E_i$ is not a stability matrix, the term $N_iG_i\xi_i(t)$ may not vanish. One can then construct $\bar{p}$ such that $x_i(t)$ does not converge to a best response of $\bar{p}$.

\begin{Example}\label{ex:noABR}
Consider the two-player coordination game
$$ R_1\left(x_1, x_2\right)=x_1\tr \Pmatrix{
1 & 0 \\
0 & 1}x_2, \quad  R_2\left(x_2, x_1\right)=x_2\tr \Pmatrix{
1 & 0 \\
0 & 1}x_1$$
which has a mixed-strategy NE at
$$x_1^*=x_2^*=\left(\frac{1}{2},\frac{1}{2}\right).$$

Suppose the dynamics of the players take the form
\begin{align*}
\dot{x}_{1}&=\diag{p_1+\phi_{1}(p_1,z_1)-\left( x_{1}\tr \left (p_{1}+\phi_{1}(p_1,z_1)\right)   \right)\one}x_{1}\\
\dot{z}_1&= -(N_1\tr p_1 - z_1) \\
&\quad \phi_{1}(p_1,z_1) = -\gamma_1 N_1(N_1\tr p_1 - z_{1}),\\
\dot{x}_{2}&=\diag{p_2+\phi_{2}(p_2,z_2)-\left( x_{2}\tr \left (p_{2}+\phi_{2}(p_2,z_2)\right) \right)\one} x_{2}\\
\dot{z}_2&= +\gamma_2(N_2\tr p_2-z_2)\\
&\quad \phi_{2}(p_2,z_2) = +\gamma_2 N_2 ( N_2\tr p_2-z_2),
\end{align*}
where $\gamma_1=20$ and $\gamma_2=50$.

\begin{figure}
\begin{center}
\includegraphics[width=0.5\textwidth]{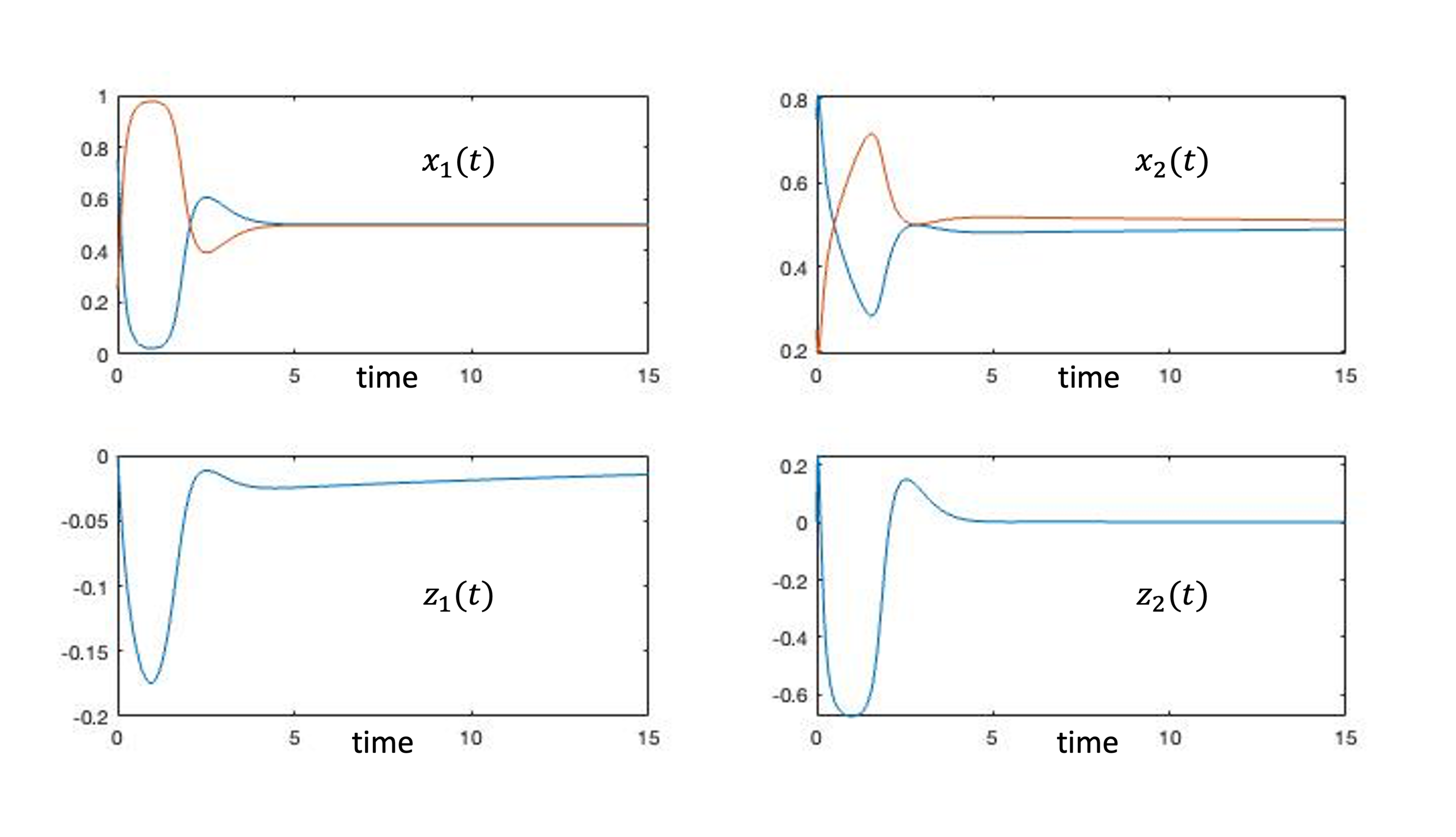}
\end{center}
\caption{Stable outcome of mixed-strategy equilibrium under higher-order replicator dynamics.}
\label{fig:stableHORD}
\end{figure}

When $p_1=P_1(x_2)$ and  $p_2=P_2(x_1)$, one can verify that the dynamics are locally stable around the mixed-strategy NE of the coordination game. Figure~\ref{fig:stableHORD} shows a representative simulation starting from $x_1(0) = x_2(0) = \Pmatrix{0.75\\0.25}$ and $z_1(0) = z_2(0) = 0$.

Now, let us inspect the first player's dynamics \textit{in isolation} with $p_1(t)$ constant. These dynamics do not play a best response, neither to $p_1=\Pmatrix{1&0}\tr$ nor to $p_1=\Pmatrix{0&1}\tr$, starting from $z_1(0)=0$ and any $x_1(0)$ in the interior of the simplex. The solution of $z_1(t)$ starting from $z_1(0)=0$  takes the form
$$z_1(t)=-\int_{0}^{t} e^{(t-\tau)} N_1\tr p_1(\tau) d \tau.$$
Let
$$N_1=\Pmatrix{-\frac{1}{\sqrt{2}}\\\frac{1}{\sqrt{2}}},$$
and inspect the solution when $p_1=\Pmatrix{1&0}\tr$. We have
$$z_1(t) \rightarrow \infty,$$
therefore
$$\phi_1(t) \rightarrow \Pmatrix{-\infty \\ \infty}\quad \text{and} \quad \left(\phi_1(t)+\Pmatrix{1\\0}\right)  \rightarrow \Pmatrix{-\infty \\ \infty}.$$
Accordingly,
$$x_1(t) \rightarrow \Pmatrix{0\\1}.$$
Figure~\ref{fig:UnstableRD} presents a simulation of player 1 responding to $p_1=\Pmatrix{1&0}\tr$, starting from various initial strategies. In all cases, $x_1(t)$ is not an asymptotic best response to the constant $p_1$.

\end{Example}

\begin{figure}[!t]
\centerline{\includegraphics[width=.3\columnwidth]{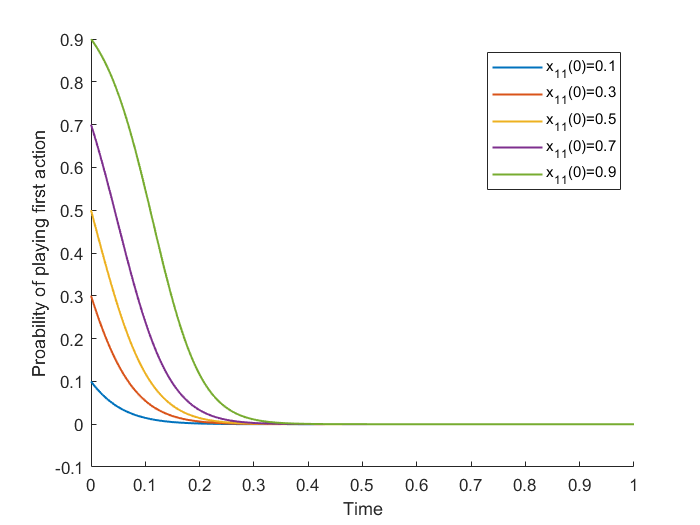}}
\caption{Dynamics of player 1 responding to $p_1=\Pmatrix{1\\0}$ from various initial strategies}
\label{fig:UnstableRD}
\end{figure}

\subsection{Learnability using internally stable higher-order components}

To recap, we have shown that the ABR property of higher-order dynamics is inherited from its standard-order counterpart whenever the higher-order are internally stable (Proposition~\ref{Prop:Internallystable}). Furthermore, in the absence of internal stability, it is possible that higher-order dynamics will not have the ABR property (Example~\ref{ex:noABR}). We now investigate the implications of imposing internal stability of the higher-order terms.

Towards this end, we will use the control theoretic concept of strong stabilization to investigate learnability using internally stable higher-order components, as before using the reduced-order linearized dynamics \eqref{eq:DecentralizedFormulationHeterogeneous}.  A system is ``strongly stabilizable'' if it can be made stable through feedback with another internally stable system. The ``parity interlacing principle" provides a necessary and sufficient condition for the strong stabilizability of a given system \cite{youla1974singleloop}. The principle states that there should be an \emph{even} number of real poles, counting with multiplicities, between each pair of real unstable blocking zeros (see Appendix~\ref{sec:linearSystems}). The decentralized version of the principle states that there should be an \emph{even} number of real poles, counting with multiplicities, between each pair of real unstable decentralized blocking zeros \cite{Ozguler1992DecentralizedSS}.

\subsubsection{A mixed-strategy NE requiring unstable higher-order dynamics}\label{sec:notstronglystabilizable}

In this section, we establish that learning some mixed-strategy NE \textit{requires} using unstable higher-order components, i.e., strong stabilization is not possible.

Consider a two-player two-action game with an isolated completely mixed-strategy NE, $x^*$, that satisfies Assumption~\ref{Assumption:Regularity}. Our general analysis framework (Section~\ref{sec:general}) takes the form
\begin{subequations}\label{eq:twoplayerss}
\begin{align}
\Pmatrix{\dot{w} \\ \dot{v}}
&= \Pmatrix{\mathcal{N}\tr \bar{J} \mathcal{N}&0\\\mathcal{M} & -I}\Pmatrix{w\\ v} + \Pmatrix{\bar{b}_1& 0\\0 & \bar{b}_2 \\0&0\\0&0}u\\
y &= \Pmatrix{\mathcal{M}&-I}\Pmatrix{w\\ v}.
\end{align}
\end{subequations}

We will consider the following assumption on the matrix $\mathcal{N}\tr \bar{J}\mathcal{N}$.

\begin{Assumption}\label{Assumption:zerotracetwoplayer}
The matrix $\mathcal{N}\tr \bar{J}\mathcal{N}$ has zero trace, i.e.,
$$\mathcal{N}\tr \bar{J}\mathcal{N}=\Pmatrix{0&n_1\\n_2&0}. $$
\end{Assumption}
The above assumption can be viewed as a consequence of NE stationarity and uncoupledness, as discussed in the proof of Claim~\ref{Claim:JacobianStructureGamma}.

Let
$$m_{12}=N_1\tr M_{12} N_2 \quad  m_{21}=N_2\tr M_{21} N_1.$$
Recalling Assumption~\ref{Assumption:Regularity}, it must be that $m_{12}\neq0$ and $m_{21}\neq0$. Furthermore, to comply with the learnability conditions of Theorem~\ref{Thm:HeterogeneousStabilization}, we assume that  $n_i \neq 0$ and $\bar{b}_i \neq 0$ for $i=1,2$.

\begin{Proposition}\label{Prop:Incomptwo}
If $n_1n_2>0$, then $x^*$ is not strongly stabilizable.
\end{Proposition}

\begin{proof}
To check strong stabilizability (see Appendix~\ref{sec:linearSystems}), we compute
$$
\mathcal{T}(s) =\Pmatrix{\mathcal{M}&-I} \left(s I-\Pmatrix{\mathcal{N}\tr \bar{J} \mathcal{N}&0\\\mathcal{M} & -I}\right)^{-1} \Pmatrix{\bar{b}_1& 0\\0 & \bar{b}_2 \\0&0\\0&0} =\frac{s}{s+1} \mathcal{M}(s I-\mathcal{N}\tr \bar{J} \mathcal{N})^{-1}\Pmatrix{\bar{b}_1&0\\0&\bar{b}_2}.
$$
The system has a blocking zero at $s=0$ and $|s|\rightarrow \infty$. Poles are a subset of the eigenvalues of
$$ \Pmatrix{\mathcal{N}\tr \bar{J} \mathcal{N}&0\\\mathcal{M} & -I},$$
which are
$$\pm \sqrt{n_1n_2},-1,-1.$$
When $n_1n_2 >0,$ the system has a single positive real pole with multiplicity one between the two blocking zeros. Therefore, the system is not strongly stabilizable.
\end{proof}

Using the previous result, it is straightforward to see that the mixed-strategy NE of the game in Example~\ref{ex:noABR} is not strongly stabilizable under target gradient play, replicator dynamics, or a mixture of the two. Next, we consider a general result that holds for identical games, where
$$M_{12}=M_{21}.$$

\begin{Proposition}\label{Prop:IdenticalInterest}
For identical two-player two-action games, isolated mixed-strategy NE are not strongly stabilizable when both players use the same learning dynamics that satisfy  Assumption~\ref{Assumption:zerotracetwoplayer}.
\end{Proposition}

\begin{proof}
Since both players use the same dynamics and utilities are identical, we have $n_1=n_2$.
\end{proof}

\subsubsection{Stable higher-order dynamics and mixed-strategy NE of competitive games}

The previous result showed that certain mixed-strategy NE are incompatible with stable higher-order dynamics. We now show that mixed-strategy NE of a class of competitive games are compatible with stable higher-order dynamics augmented to target gradient play dynamics.

\begin{Definition} [\cite{Daskalakis2009NetworkZeroSum}]
In a zero-sum graphical polymatrix game, the utility of a player takes the form
$$R_i(x)= x_i\tr \sum_{j=1\atop j\not= i}^n M_{ij} x_j,$$
and the utility matrices satisfy the following
$$M_{ij}=-M_{ji}\tr \quad \forall \quad i, j. $$
\end{Definition}

\begin{Proposition}\label{Prop:Zerosumgrap}
An isolated completely mixed-strategy NE of a zero-sum graphical polymatrix games is strongly stabilizable under higher-order target gradient play dynamics.
\end{Proposition}

\begin{proof}
Consider the local dynamics matrix of standard-order target gradient play around $x^*$
$$ \mathcal{M}=\mathcal{N}\tr \Pmatrix{
0 & M_{12} & \hdots &  M_{1n} \\[5pt]
-M_{12}\tr & 0 & \hdots &  M_{2n} \\[5pt]
\vdots & \vdots & \ddots & \vdots\\[5pt]
-M_{1n}\tr & -M_{2n}\tr & \hdots & 0
} \mathcal{N}.$$
The matrix $\mathcal{M}$ is skew-symmetric. Therefore, all of its eigenvalues have a real part equal to zero, and the strong stabilizability condition is trivially satisfied.
\end{proof}

For two-player games, a larger class of competitive games is the class of games that are strategically equivalent to zero-sum games. We will use the following definition, adapted from \cite[Theorem 1]{Jour1978Strategically}, as a starting point.

\begin{Definition}\label{Def:Strategically}
A two-player game with payoff matrices $M_{12}$ and $M_{21}$ is strategically equivalent to a zero-sum game iff
\begin{itemize}
\item $\exists$ $\alpha>0$ such that $M_{12}-\alpha A = Q_{1}$
\item  $\exists$ $\beta>0$ such that
$M_{21}+\beta A\tr = Q_{2}$,
\end{itemize}
where $ A \in \mathbb{R}^{k_1\times k_2}$ and each $Q_{i}$ is a matrix with all equal rows.
\end{Definition}

\begin{Proposition}\label{Prop:StrategicallyZerosum}
An isolated completely mixed-strategy NE of a two-player strategically equivalent zero-sum game is strongly stabilizable under higher-order target gradient play dynamics.
\end{Proposition}

\begin{proof}
Using Definition~\ref{Def:Strategically}, the local dynamics of standard-order target gradient play dynamics around $x^*$ take the form
$$
\mathcal{M} =\mathcal{N}\tr \Pmatrix{0 & \alpha A +Q_1 \\ -\beta A\tr + Q_{2} & 0} \mathcal{N}= \mathcal{N}\tr \Pmatrix{0 & \alpha A  \\ -\beta A\tr & 0} \mathcal{N}+\underbrace{\mathcal{N}\tr \Pmatrix{0 & Q_{1} \\  Q_{2} & 0} \mathcal{N}}_{0}.
$$
Let $K= N_1\tr A N_2$, and suppose $(\lambda,v)$ is an eigenvalue/eigenvector pair of $\mathcal{M}$. Notice that $\mathcal{M}$ is non-singular since $x^*$ is isolated. This isolation necessitates that $k_1=k_2$.
Next, we have
$$\Pmatrix{0 & \alpha K  \\ -\beta K\tr  & 0} \Pmatrix{v_1\\v_2} = \lambda \Pmatrix{v_1\\v_2}$$
which gives
$$  \alpha K v_2 = \lambda v_1 \quad \And \quad  -\beta K\tr v_1 = \lambda v_2 $$
leading to
$$-\alpha \beta K\tr K v_2= \lambda^2 v_2.$$
If $\lambda$ is an eigenvalue of $\mathcal{M}$, then $\frac{\lambda^2}{\alpha \beta}$ must be an eigenvalue of $-K\tr K$. The matrix $-K\tr K$ is negative definite since $K$ is invertible. Therefore, all eigenvalues of $\mathcal{M}$ have a real part equal to zero.
\end{proof}

\subsubsection{Stable higher-order dynamics and strict NE}

Proposition~\ref{Prop:Incomptwo} presented a condition on a two-player two-action game such that converging to a mixed-strategy NE necessitated unstable higher-order dynamics. For target gradient play, the condition becomes $m_{12}m_{21}>0$, which only holds when the associated $2\times 2$ game also has a strict NE.

When players have more than two actions, and in the presence of strict NE, it is possible that unstable higher-order dynamics are not required.

\begin{Example}\label{eq:targetexample}
Consider the following coordination game
$$
R_1(x_1,x_2) = x_1\tr \Pmatrix{2 &1 & 1 \\1& 2 &1\\1&1&2} x_2, \quad R_2(x_2,x_1) = x_2\tr \Pmatrix{2 &3 & 2 \\1& 4 &3 \\1&2&4} x_1.
$$
This game is an ordinal potential game with $7$ Nash equilibria: three strict, three mixed, and one completely mixed-strategy NE, $x^*$, at
$$ x_1^*=\left(\frac{1}{2}, \frac{1}{6}, \frac{1}{3}\right) \quad \And \quad  x_2^*=\left(\frac{1}{3}, \frac{1}{3}, \frac{1}{3}\right).$$ The eigenvalues of the local dynamics matrix of target gradient play around $x^*$ are
$\pm 1.1721 \pm 0.2011j.$
The completely mixed-strategy NE is strongly stabilizable when using target gradient play dynamics. In particular, standard anticipatory target gradient play dynamics \cite{arslan2006anticipatory} can locally lead to $x^*$, albeit with a relatively small region of attraction.
\end{Example}

The behavior of anticipatory learning dynamics in Example~\ref{eq:targetexample} is associated with another desirable property of learning dynamics, namely that they exhibit local stability in the neighborhood of strict NE.

To this end, consider the following version of higher-order target gradient play dynamics
\begin{subequations}\label{eq:higherorderGP}
\begin{align}
\dot{x}_{i}&=\Pi_{\Delta}[x_i+P_i(x_{-i})+\phi_i(P_i(x_{-i}),v_i,\xi_i)]-x_i\\
\dot{v}_i&= N_i\tr P_i(x_{-i})-v_i, \label{eq:higherorderRD3}\\
\dot{\xi}_i&=E_i\xi_i+F_i(N_i\tr P_i(x_{-i})-v_i)\\
&\quad \phi_{i}(P_i(x_{-i}),v_i,\xi_i) =N_i( G_{i} \xi_{i}+H_{i} (N_i\tr P_i(x_{-i})-v_i)),
\end{align}
\end{subequations}
where the $x_{-i}$ also evolve according to such variants of higher-order target gradient play.

\begin{Proposition}
Consider a polymatrix game with a strict NE $x^*$. If all players use higher-order target gradient play dynamics of the form \eqref{eq:higherorderGP} and the higher-order components of all players are internally stable, then $(x^*,v^*,0)$ is locally asymptotically stable.
\end{Proposition}

\begin{proof} (sketch)
Consider the Lyapunov function:
\begin{align*}
V(x,\xi,v)=\sum_i \frac{\alpha_i}{2} (x_i-x_i^*)\tr(x_i-x_i^*)+ \sum_i \gamma_i \xi_i\tr B_i \xi_i + \sum_i \frac{\beta_i}{2} (N_i\tr P_i(x_{-i})-v_i)\tr(N_i\tr P_i(x_{-i})-v_i),
\end{align*}
for appropriately chosen positive constants $\alpha_i$, $\beta_i$, and $\gamma_i$. The matrix $B_i$ is the unique positive definite solution of
$$B_iE_i+E_i\tr B_i= -\frac{1}{\gamma_i}I,$$

One can show that if
$$V(x,\xi,v)< \delta $$
for a sufficiently small $\delta$, then
$$\dot{V} < 0.$$
In the proof, we use the following property of convex projections \cite{Aubin1984}
$$(\Pi_\Delta[x]-s)\tr(\Pi_\Delta[x]-x) \leq 0 \quad \forall s \in \Delta,$$
the fact that $x^*$ is a strict NE, and that $\phi_i$ is zero at equilibrium.
\end{proof}

\section{Learning mixed-strategy NE in the bandit setting} \label{Sec:Bandit}\ \\[-15pt]

We now address learnability of mixed-strategy NE in a setup where players receive only instantaneously realized utilities. Our previous setup of deterministic and continuous-time dynamics assumed that each player has access to a payoff vector, $P_i(x_{-i})$, that represented the expected payoff of each discrete action.
We now consider a discrete-time stochastic setup, where each player, $i$, uses their mixed strategy to play an action $a_i(t)$. Each player then observes $r_i(a(t))$, where $a(t)\in \mathcal{A}$ is the action profile played at time $t$. Players know nothing about the structure of the game, including their own utility functions.

The main result is as follows:

{\textsc{Theorem (informal)}. \textit{For any finite game with an isolated completely mixed-strategy NE $x^*$ that satisfies Assumption~\ref{Assumption:Regularity}, there exists a higher-order bandit-based algorithm under which players' strategies converge to $x^*$ with a positive probability.}

We will first define a bandit version of our previously introduced higher-order replicator dynamics. The main tool is the ODE method of stochastic approximation (Proposition~\ref{Prop:SI}), in which asymptotic properties of discrete-time stochastic iterations can be analyzed through continuous-time deterministic differential equations. The main idea is to construct discrete-time iterations that, when converted to a differential equation, fall in the framework of Theorem~\ref{Thm:RDStabilization}. The same approach was utilized to study learnability of mixed-strategy NE in polymatrix games using higher-order gradient play dynamics \cite{toonsi2024bandit}. 

\subsection{Higher-order bandit-based algorithm and associated mean dynamics}

We make the following simplifying assumption throughout Section~\ref{Sec:Bandit}.

\begin{Assumption}\label{Assumption:PositiveUtilities}
Utility functions of all players are strictly positive, i.e.,
$$r_i(a)>0 \quad \quad \forall  a \in \mathcal{A} \quad \text{and} \quad i \in \mathcal{I}.$$
\end{Assumption}
For any finite game, a constant bias term can be added to each player's utility to satisfy the above assumption.
Adding such a constant does not change the strategic nature of the game.

We first present a bandit-based learning algorithm related to standard-order replicator dynamics. For the $i^\mathrm{th}$ player, these are:
\begin{subequations}\label{eq:RDBanditStandarda}
\begin{align}
x_i(t+1) &= x_i(t)+\frac{\delta_i r_i(a(t))}{t+1}[\mathbf{e}_{a_i(t)}-x_i(t)],\\
\sigma_i(t) &=\Pi_{\Delta_{\epsilon}}[x_i(t)]
\end{align}
\end{subequations}
where the randomized action $a_i(t)$ is chosen according to
\begin{equation}\label{eq:RDBanditStandardb}
\prob{a_i(t) = a_i' \in \mathcal{A}_i} = \sigma_{i(a_i')}(t).\\
\end{equation}
In these equations, $\sigma_{i(a_i')}(t)$ denotes the $(a_i')^\mathrm{th}$ component of $\sigma_i(t)$.
The vector $\mathbf{e}_{a_i(t)}$ is the unit vector, as defined in \eqref{eq:DefinitionOfe}, which equals 1 at the $a_i^\mathrm{th}(t)$ position. The variable $\delta_i>0$ is a step size that can be chosen small enough so that iterations remain in the simplex, i.e.,
$$x_i(0)\in \Delta(k_i) \implies  x_i(t)\in \Delta(k_i) \quad \forall t.$$
The strategy, $\sigma_i(t)$, is the projection of $x_i(t)$ onto the set
$$\Delta_{\epsilon}(k_i) = \theset{v \in \mathbb{R}^{k_i} \st v_j\ge \epsilon, j = 1, ..., k_i, \und \sum_{j=1}^{k_i} v_j = 1}.$$
With this change, the vector $x_i$ can be interpreted as the propensities of player $i$ to play each action, with the projection to $\Delta_\epsilon$ encouraging exploration, since any action will be played with probability at least $\epsilon$. When $x_i(t)$ lies deeper in the interior of the simplex, then the projection is not binding and $\sigma_i(t) = x_i(t)$.

The learning algorithm in \eqref{eq:RDBanditStandarda}--\eqref{eq:RDBanditStandardb} in the absence of projection to $\Delta_\epsilon$ has been widely studied in the literature, e.g., \cite{narendra2012learning,Verbeeck2002learningPareto}. Its connection to replicator dynamics has been established in various works such as \cite{BORGERS19971RDreinfrocement,HOPKINS2005110RDreinfrocemen2}.
We now derive the connection between \eqref{eq:RDBanditStandarda}--\eqref{eq:RDBanditStandardb} and a modified version of replicator dynamics. First, note that
$$\expected{r_i(a(t)) \mathbf{e}_{a_i(t)}\given x(t)}= \diag{P_i(\sigma_{-i}(t))} \sigma_i(t)$$
and
$$ \expected{r_i(a(t))  x_i(t) \given x(t) }=(P_i(\sigma_{-i}(t))\tr \sigma_i(t)) x_i(t).$$
It follows that the mean dynamics of \eqref{eq:RDBanditStandarda}--\eqref{eq:RDBanditStandardb} (see Appendix~\ref{sec:ODEMethod}) are
\begin{subequations}\label{eq:ModifiedRD}
\begin{align}
\dot{x}_i&=\diag{P_i(\sigma_{-i})}\sigma_{i}-\diag{(\sigma_i\tr P_i(\sigma_{-i}))\one} x_i\\
\sigma_i&=\Pi_{\Delta_{\epsilon}}[x_i].
\end{align}
\end{subequations}
When $\epsilon=0$, we retrieve replicator dynamics. Furthermore, these equations coincide with standard replicator dynamics around any isolated completely mixed-strategy NE (for sufficiently small $\epsilon$).

We now construct a higher-order version of \eqref{eq:RDBanditStandarda}\footnote{See \cite{Chasparis2012DD} for a related approach, where the interest was in equilibrium selection of pure strategy NE, rather than stabilization of mixed-strategy NE.}. The dynamics of the $i^\mathrm{th}$ player are:
\begin{subequations}\label{eq:RDBanditHigher}
\begin{align}
\label{eq:iteratesofHORD1} x_i(t+1)&=x_i(t)+\frac{\delta_i }{t+1}\left( r_i(a(t)) \left(\mathbf{e}_{a_i(t)}-x_i(t)\right)+ \diag {\Tilde{\phi}_i(t)-(x_i(t)\tr \Tilde{\phi}_i(t)) \one} x_i(t)\right)\\
v_i(t+1)&=v_i(t)+\frac{\delta_i}{t+1} (N_i\tr\Tilde{P}_i(t)-v_i(t))\\
\xi_i(t+1) & = \xi_i(t)+\frac{\delta_i}{t+1}(E_i \xi_i(t) + F_i(N_i\tr \Tilde{P}_i(t)-v_i(t))),\\
\sigma_i(t) &=\Pi_{\Delta_{\epsilon}}[x_i(t)]
\end{align}
\end{subequations}
where
\begin{align*}
\Tilde{P}_i(t)=\frac{r_i(a(t))}{\sigma_{i(a_i(t))}(t)} \mathbf{e}_{a_i(t)},\quad \Tilde{\phi}_{i}(t) =N_i( G_{i} \xi_{i}(t)+H_{i} (N_i\tr \tilde{P}_i(t)-v_i(t))).
\end{align*}
As before, the randomized action, $a_i(t)$, is chosen according to \eqref{eq:RDBanditStandardb}.

We are now ready to build a connection between \eqref{eq:RDBanditHigher} and our higher-order version of replicator dynamics. Straightforward calculations show that
\begin{align*}
\expected{ N_i\tr \Tilde{P}_i(t) \given x(t),\xi(t),v(t)}&= N_i\tr P_i(\sigma_{-i}(t)),\\
\expected{\diag{\Tilde{\phi}_i(t)}x_i|x(t),\xi(t),v(t)}&= \diag{\phi_i(t)}x_i ,\\
\expected{\Tilde{\phi}_i(t)\tr x_i|x(t),\xi(t),v(t)}&=\phi_i(t)\tr x_i,
\end{align*}
where
$$\phi_i(t)=N_i( G_{i} \xi_{i}(t)+H_{i} (N_i\tr P_i(\sigma_{-i}(t))-v_i(t))).$$
Accordingly, the relevant mean dynamics of \eqref{eq:RDBanditHigher} are
\begin{subequations}\label{eq:ModifiedHigherorderRD}
\begin{align}
\dot{x}_i&=\diag{ P_i(\sigma_{-i})} \sigma_i +\diag{\phi_i} x_i -\diag{(\sigma_i\tr P_i(\sigma_{-i})+x_i\tr\phi_{i})\one}x_i\\
\dot{v}_i&= N_i\tr P_i(\sigma_{-i}) -v_i\\
\dot{\xi}_i&=E_i\xi_i+F_i(N_i\tr  P_i(\sigma_{-i}) -v_i),\\
\sigma_i &=\Pi_{\Delta_\epsilon}\left[x_i(t)  \right]
\end{align}
\end{subequations}
where
$$\phi_i=N_i( G_{i} \xi_{i}+H_{i} (N_i\tr P_i(\sigma_{-i})-v_i)),$$
and the $x_{-i}$ and $\sigma_{-i}$ also evolve according to \eqref{eq:ModifiedHigherorderRD}.

The connection to higher-order replicator dynamics in (\ref{eq:higherorderRD}) is clear. In particular, when the projection to $\Delta_\epsilon$ is not binding, \eqref{eq:ModifiedHigherorderRD} coincides with \eqref{eq:higherorderRD}.

\subsection{Case I: Internally stable higher-order components}

We are now in a position to state a main result regarding the bandit case and utilize the methods of Appendix~\ref{sec:ODEMethod}.

\begin{Theorem}\label{Thm:BanditCaseI}
In the framework of Theorem~\ref{Thm:RDStabilization}, let the equilibrium $(x^*,0,0)$ be locally exponentially stable under higher-order replicator dynamics \eqref{eq:closedloopHORD}--\eqref{eq:controller}, where $E_i$ is a stability matrix for each $i=1,...,n$. Then under bandit-based replicator dynamics \eqref{eq:RDBanditHigher} with sufficiently small $\epsilon > 0$, sufficiently small $\delta_i>0$ for each $i=1,...,n$, and randomized actions selected according to \eqref{eq:RDBanditStandardb}, the iterations $(x(t),v(t),\xi(t))$ converge to $(x^*,0,0)$ with positive probability.
\end{Theorem}

The remainder of this subsection is devoted to the proof of Theorem~\ref{Thm:BanditCaseI}.

The desired result will follow from Proposition~\ref{Prop:SI}. We already have established that the relevant mean dynamics of \eqref{eq:RDBanditHigher} corresponding to \eqref{eq:Mean} are \eqref{eq:closedloopHORD}--\eqref{eq:controller} for equilibria that are in the interior of $\Delta_\epsilon$. Conditions \#1, \#3, and \#4 are straightforward. Attainability (Condition \#5) is assured by the projection to $\Delta_\epsilon$. It remains to verify that the iterates are bounded (Condition \#2). It is for this reason that we have separated the case where the $E_i$ are stability matrices.

\begin{Claim}
For initial conditions $x_i(0)\in \Delta(k_i)$ and $(v_i(0),\xi_i(0))$, there exists $\delta_i$ sufficiently small so that the iterates evolve in a compact set and $x_i(t)$ remains in the simplex.
\end{Claim}

\begin{proof}
First, note that $\Tilde{P}_i(t)$ is uniformly bounded over time. Suppose $\delta_i=1$.
It is direct to show that  $v_i(t)$ and consequently $F_i(N_i\tr \Tilde{P}_i(t)-v_i(t)))$ are uniformly bounded over time.
One can then consider a bounded set on which $v_i$ evolves and take its closure to show that $v_i$ evolves on a compact set.

Next, use the following Lyapunov function to show boundedness of $\xi_i$
$$ V_i= \xi_i\tr B_i \xi_i,$$
where each $B_i$ is the unique positive definite solution to
$$B_iE_i+E_i\tr B_i=-I.$$
Therefore, $\xi_i$ also evolves on a compact set. The bounds that hold on $\xi_i$ and $v_i$ when $\delta_i=1$ also hold for any $0<\delta_i<1$.

Next, we show that  $\delta_i$ can be chosen sufficiently small so that $x_i$ stays on the simplex. It is direct to see that $\one\tr x_i(t)=1$ for all $t$. It remains to show that if $x_i(t)>0$ entry-wise, then $x_i(t+1)> 0$ entry-wise. Let $x_{ik}$ denote the $k^{\text{th}}$ component of $x_i$.

For $k \neq a_i(t)$, we have
$$x_{ik}(t+1) = x_{ik}(t)+\frac{\delta_i}{t+1}\left(r_i(a(t)) (-x_{ik}(t))+\Tilde{\phi}_{ik}(t) x_{ik}(t) -(\Tilde{\phi}_i(t)\tr x_i(t)) x_{ik}(t)\right),$$
so we need
$$ x_{ik}(t)+\frac{\delta_i}{t+1}\left(r_i(a(t)) (-x_{ik}(t))+\Tilde{\phi}_{ik}(t) x_{ik}(t) -(\Tilde{\phi}_i(t)\tr x_i(t)) x_{ik}(t)\right)> 0.$$
Neglect $\frac{1}{t+1}$ and divide over $x_{ik}$ to get
$$ \delta_i\left(-r_i(a(t)) +\Tilde{\phi}_{ik}(t) -(\Tilde{\phi}_i(t)\tr x_i(t))\right)> -1.$$
If
$$\left(-r_i(a(t)) +\Tilde{\phi}_{ik}(t) -(\Tilde{\phi}_i(t)\tr x_i(t))\right)<0$$
we need
$$ \delta_i<\frac{1}{\left(r_i(a(t)) -\Tilde{\phi}_{ik}(t) +(\Tilde{\phi}_i(t)\tr x_i(t))\right)}.$$
The function $\Tilde{\phi}_i(t)$ is bounded given the bounds on $\tilde{P}_i(t)$, $v_i(t)$, and $\xi_i(t)$. The term $\tilde{\phi}_i\tr x_i$ is an average of elements of $\tilde{\phi}_i$, so it is also bounded.

For $k=a_i(t)$, to have
$$x_{ik}(t+1)>0$$
we need
$$ \delta_i\left(\underbrace{r_i(a(t)) (\frac{1}{x_{ik}}-1)}_{\text{positive}} +\Tilde{\phi}_{ik}(t) -(\Tilde{\phi}_i(t)\tr x_i(t))\right)>-1.$$
Therefore, it is sufficient to have
$$ \delta_i\left(\Tilde{\phi}_{ik}(t) -(\Tilde{\phi}_i(t)\tr x_i(t))\right)>-1.$$
When
$$\left(\Tilde{\phi}_{ik}(t) -(\Tilde{\phi}_i(t)\tr x_i(t))\right)<0$$
we need
$$ \delta_i<\frac{1}{-\left(\Tilde{\phi}_{ik}(t) -(\Tilde{\phi}_i(t)\tr x_i(t))\right)}.$$
Hence, we can derive a bound on $\delta_i$ to keep $x_i$ on the simplex.
\end{proof}

\subsection{Case II: Internally unstable higher-order components}

We saw in Section~\ref{sec:notstronglystabilizable} that it may be the case where it is necessary that at least one of the $E_i$ in the higher-order dynamics is not a stability matrix. In such a situation, the boundedness of the iterates in  \eqref{eq:RDBanditHigher} is not guaranteed.

To address this situation, we introduce the following modification of \eqref{eq:RDBanditHigher}:
\begin{subequations}\label{eq:InternallyUnstableIterates}
\begin{align}
\label{eq:iteratesofHORD1} x_i(t+1)&=x_i(t)+\frac{\delta_i }{t+1}\left( r_i(a(t)) \left(\mathbf{e}_{a_i(t)}-x_i(t)\right)+ \diag {\Tilde{\phi}_i(t)-(x_i(t)\tr \Tilde{\phi}_i(t)) \one} x_i(t)\right)\\
v_i(t+1)&=v_i(t)+\frac{\delta_i}{t+1} (N_i\tr\Tilde{P}_i(t)-v_i(t))\\
\xi_i(t+1) & = \xi_i(t)+\frac{\delta_i}{t+1}(E_i \xi_i(t) + F_i(N_i\tr \Tilde{P}_i(t)-v_i(t))+\psi_i(t)),\\
\sigma_i(t) &=\Pi_{\Delta_{\epsilon}}[x_i(t)]
\end{align}
\end{subequations}
where
$$ \Tilde{P}_i(t)=\frac{r_i(a(t))}{\sigma_{i(a_i(t))}(t)} \mathbf{e}_{a_i(t)},\quad \Tilde{\phi}_{i}(t) =N_i( G_{i} \xi_{i}(t)+H_{i} (N_i\tr \Tilde{P}_i(t)-v_i(t))).
$$

These dynamics have a newly introduced term, $\psi_i(t)$, defined as
$$ \psi_i(t)=-m_i(\xi_i(t)) \cdot \xi_i(t) +\Omega_i(t),$$
where
$$m_i(\xi_i(t))=\min(\beta_i,\norm{\xi_i(t)}^2).$$
Here, $\beta_i > 0$  is a sufficiently large constant, and $\Omega_i(t)$ is a
random vector with each component being a uniformly distributed zero-mean i.i.d. random variable with suitable support. If $E_i$ is a stability matrix, the extra term is unnecessary and one can set $\psi_i(t) \equiv 0$.
The effect of $m_i(\xi_i(t))$ is that its quadratic growth will enforce that the $\xi_i(t)$ iterates remain bounded.

We can now state a modified main result regarding the bandit case to address this situation.

\begin{Theorem}\label{Thm:BanditCaseII}
In the framework of Theorem~\ref{Thm:RDStabilization}, let the equilibrium $(x^*,0,0)$ be locally exponentially stable under higher-order replicator dynamics \eqref{eq:closedloopHORD}--\eqref{eq:controller}. Then under modified bandit-based replicator dynamics \eqref{eq:InternallyUnstableIterates} with sufficiently small $\epsilon > 0$, sufficiently small $\delta_i>0$ for each $i=1,...,n$, and randomized actions selected according to \eqref{eq:RDBanditStandardb}, the iterations $(x(t),v(t),\xi(t))$ converge to $(x^*,0,0)$ with positive probability.
\end{Theorem}

The remainder of this section is devoted to the proof of Theorem~\ref{Thm:BanditCaseII}.

Similarly to the previous case, the relevant mean dynamics of \eqref{eq:InternallyUnstableIterates} are
\begin{subequations}\label{eq:InternallyUnstableDynamics}
\begin{align}
\dot{x}_i&=\diag{ P_i(\sigma_{-i})} \sigma_i +\diag{\phi_i} x_i -\diag{(\sigma_i\tr P_i(\sigma_{-i})+x_i\tr\phi_{i})\one}x_i\\
\dot{v}_i&= N_i\tr P_i(\sigma_{-i}) -v_i\\
\dot{\xi}_i&=E_i\xi_i-m_i(\xi_i) \xi_i+F_i(N_i\tr  P_i(\sigma_{-i}) -v_i),\\
\sigma_i &=\Pi_{\Delta_\epsilon}\left[x_i(t)  \right]
\end{align}
\end{subequations}
where
$$\phi_i=N_i( G_{i} \xi_{i}+H_{i} (N_i\tr P_i(\sigma_{-i})-v_i)),$$
For any $E_i$ that is a stability matrix, we set $m_i(\xi_i)=0$. Otherwise, because of the quadratic structure, this term does not impact linearization analysis at $\xi= 0$.

It remains to verify attainability (Condition \#5) and the boundedness of iterates (Condition \#2) for Proposition~\ref{Prop:SI}.

\begin{Claim}\label{Prop:AttainabilityAndBoundedness}
For the iterates in \eqref{eq:InternallyUnstableIterates}, suppose $E_i$ is not a stability matrix.
Let $\beta_i$ be such that
$$E_i+E_i\tr-2\beta_iI$$
is a stability matrix, and $\eta_i$ be such that $$\norm{ F_i(N_i\tr \Tilde{P}_i(t)-v_i(t))}<\eta_i \quad \forall t.$$
Then, there exists $n_{\mathrm{max}}$ that assures $\xi_i$ is bounded and that $0$ is attainable by $\xi_i$, where
$$\norm{\Omega_i(t)}<n_{\mathrm{max}}\quad \forall t.$$
\end{Claim}

\begin{proof}
First, note that whenever $\norm{\xi_i(t)}^2 \leq \beta_i$, we have
\begin{align*}
\norm{(E_i-\norm{\xi_i}^2I)\xi_i}\leq \underbrace{\norm{E_i} \sqrt{\beta_i} +\beta_i \sqrt{\beta_i}}_{\text{c}}.
\end{align*}

Next, let each entry of $\Omega_i(t)$ be a uniform random variable on some interval $[-L,L]$, where $L=c+\eta_i+\sqrt{\beta_i}$. Define $y_i= F_i(N_i\tr \Tilde{P}_i(t)-v_i(t))$.
When $\norm{\xi_i(t)}^2>\beta_i$, for any arbitrarily small $\epsilon$ we have
$$\mathbf{Prob} \left[\Omega_i(t) \in (-F_iy_i(t)-\epsilon \mathbf{1}, -F_iy_i(t)+\epsilon \mathbf{1})\right]>0.$$
Consequently, we can have
$$V_i(\xi_i(t+1))-V_i(\xi_i(t))<0 \quad \text{for large } t, $$
where $V_i(\xi_i)=\xi_i\tr\xi_i$. Consider now the case when $\norm{\xi_i(t)}^2\leq \beta_i$. For any arbitrarily small $\epsilon$ we have

$$\mathbf{Prob} \left[\Omega_i(t) \in \left( T_i-\epsilon \mathbf{1},T_i+\epsilon \mathbf{1}\right)\right]>0,$$
where 
$$T_i=-F_iy_i(t)-(E_i-\norm{\xi_i(t)}^2I+I)\xi_i(t).$$
Since $\epsilon$ can be arbitrarily small, we get that for any $t>0$ and every open neighborhood $U$ of $0$
$$\mathbf{Prob}[\exists s \geq t : \xi_i(s) \in U]>0.$$
Finally, to prove boundedness, use the Lyapunov function $V_i(\xi_i)=\xi_i\tr\xi_i$ to show that for large enough $t$, $\beta_i$, and $\xi_i$ we have $V_i(\xi_i(t+1))-V_i(\xi_i(t))<0.$
\end{proof}

With this modification, the dynamics in~\eqref{eq:InternallyUnstableDynamics} still have $(x^*,0,0)$ as an equilibrium, and the local dynamics near this equilibrium are not affected by the $m_i(\xi_i(t))$ term of any player.
However, the dynamics could potentially have other equilibria. Since our results are local and we seek a positive probability of convergence, such a modification still serves our purpose.

\section{Conclusion} \label{Sec:Conclusion}\ \\[-15pt]

We established learnability of mixed-strategy NE by uncoupled higher-order replicator dynamics. Specifically, for any regular completely mixed-strategy NE of a finite game, we showed the existence of higher-order replicator dynamics that converge to this equilibrium. The main takeaway is that impossibility of learning mixed-strategy NE arises from the \textit{combination} of uncoupled dynamics and restriction to standard-order dynamics. By relaxing the latter constraint, convergence to mixed-strategy NE becomes possible.

Our modular perspective enabled several extensions. We established learnability under heterogeneous dynamics, where players use different learning dynamics. We also extended these results to bandit settings, where players observe only the realized payoffs. We further showed that the possibilities under higher-order learning do not imply universality, as there are no higher-order learning dynamics that learn mixed-strategy NE across all games.

An enabling perspective throughout is to view learning dynamics as open systems mapping payoff signals to strategic evolution. This viewpoint reveals connections between learning in games and feedback stabilization. In particular, the existence of learning dynamics leading to mixed-strategy NE corresponded to decentralized stabilization. As higher-order learning dynamics are not readily interpretable in a traditional sense in learning in games,  we introduced the asymptotic best-response (ABR) property. We analyzed the relationship between ABR and the internal stability of the higher-order components and studied compatibility of various mixed-strategy NE with this property. Again, the connection to feedback stabilization, specifically to the concept of strong stabilization, proved useful.

The broader relevance in higher-order learning is its ability to model increasingly sophisticated learning agents, particularly as ``multi-agent AI'' becomes more prominent (cf., \cite{hammond2025multiagent} for an extended discussion). Future work may explore the implications of additional structural constraints on higher-order learning, potentially leading to new limitations (e.g., \cite{fox2013population, abdelraouf2025passivity}).

\appendix

\section{Illustration:  Three players with two actions}\label{App:A0}\ \\[-15pt]

Consider a three-player two-action game. The mixed strategies of the three players are
$$x_1 = \Pmatrix{x_{11}\\ x_{12}},\quad x_2 = \Pmatrix{x_{21}\\ x_{22}}\und x_3 = \Pmatrix{x_{31}\\ x_{32}}.$$
The payoff vector of player 1 is
$$
P_1(x_{-1})= \Pmatrix{r_{111} x_{21}x_{31}+ r_{112} x_{21}x_{32}+ r_{121} x_{22}x_{31}+ r_{122} x_{22}x_{32}\\ r_{211} x_{21}x_{31}+
r_{212} x_{21}x_{32}+ r_{221} x_{22}x_{31}+ r_{222} x_{22}x_{32}},
$$
where we use the shorthand notation $r_{abc}$ for $r_1(a,b,c)$. 

Then 
$$\nabla_{x_2} P_1(x^*_{-1}) = M_{12} = 
\Pmatrix{r_{111}x_{31}^*+r_{112}x_{32}^*&\quad&r_{121}x_{31}^*+r_{122}x_{32}^*\\
r_{211}x_{31}^*+r_{212}x_{32}^*&\quad&r_{221}x_{31}^*+r_{222}x_{32}^*}$$
$$\nabla_{x_3} P_1(x^*_{-1}) = M_{13} = 
\Pmatrix{r_{111}x_{21}^*+r_{121}x_{22}^*&\quad& r_{112}x_{21}^*+r_{122}x_{22}^*\\
r_{211}x_{21}^* + r_{221}x_{22}^*&\quad&r_{212}x_{21}^*+r_{222}x_{22}^*}$$

Define
$$x_i=x_i^*+Nw_i \quad \text{for }i=1,2,3,$$
where $N\tr=\Pmatrix{\frac{1}{\sqrt{2}} & -\frac{1}{\sqrt{2}}}.$
The payoff vector of player 1 in the vicinity of a completely mixed-strategy NE, $x^*$, takes the form  
 \begin{align*}
 P_1&(x_2^*+Nw_2,x_3^*+Nw_3)\\
 &=
 \Pmatrix{
 r_{111} (x^*_{21}+\frac{1}{\sqrt{2}} w_2)(x^*_{31}+\frac{1}{\sqrt{2}} w_3)\\ 
 r_{211}(x^*_{21}+\frac{1}{\sqrt{2}} w_2)(x^*_{31}+\frac{1}{\sqrt{2}} w_3)} +
 \Pmatrix{
r_{112}(x^*_{21}+\frac{1}{\sqrt{2}} w_2)(x^*_{32}-\frac{1}{\sqrt{2}} w_3)\\ 
r_{212}(x^*_{21}+\frac{1}{\sqrt{2}} w_2)(x^*_{32}-\frac{1}{\sqrt{2}} w_3)}\\
&+
\Pmatrix{r_{121} (x^*_{22}-\frac{1}{\sqrt{2}} w_2)(x^*_{31}+\frac{1}{\sqrt{2}} w_3)\\ 
r_{221} (x^*_{22}-\frac{1}{\sqrt{2}} w_2)(x^*_{31}+\frac{1}{\sqrt{2}} w_3)} +
\Pmatrix{r_{122} (x^*_{22}-\frac{1}{\sqrt{2}} w_2)(x^*_{32}-\frac{1}{\sqrt{2}} w_3)\\
r_{222} (x^*_{22}-\frac{1}{\sqrt{2}} w_2)(x^*_{32}-\frac{1}{\sqrt{2}} w_3)}.
\end{align*} 
Rearranging terms, we have
\begin{align*}
P_1(x_2^*+Nw_2,x_3^*+Nw_3)=P_1(x_{-1}^*)+M_1 \mathcal{N} w+\Tilde{P}_{1}(w_{-1}),
\end{align*}
where
$$\mathcal{N} = \bdiag{N,N,N}, \quad M_1 = \Pmatrix{0_{2\times 2}&M_{12}&M_{13}},$$ 
and 
$$\Tilde{P}_{1}(w_{-1})=\frac{1}{2}\Pmatrix{r_{111}-r_{112}-r_{121}+r_{122} \\r_{211}-r_{212}-r_{221}+r_{222}}w_2w_3.$$
When $x^*$ is a completely mixed-strategy NE,
$$P_1(x^*_1) = \alpha_1 \one $$
for some constant $\alpha_1$, and so $N\tr P_1(x^*_1) = 0$.  The reduced-order payoff vector becomes 
\begin{align*}
 N\tr P_1(x_2^*+Nw_2,x_3^*+Nw_3)= 
\underbrace{ N\tr M_1 \mathcal{N}}_{\mathcal{M}_1} w + N\tr \Tilde{P}_{1}(w_{-1}),
\end{align*} 
where $\mathcal{M}_1$ is $1\times 3$. The full $3\times 3$ matrix
 $$\mathcal{M} = \Pmatrix{\mathcal{M}_1\\ \mathcal{M}_2\\ \mathcal{M}_3}$$ 
 can be constructed by going through the similar steps for players 2 and 3. 

\section{Background material}\ \\[-15pt]

The notation of each of these sections is independent from the remainder of the paper.

\subsection{Local stabilization of nonlinear dynamical systems}\label{sec:ofStabilization}

Consider the nonlinear dynamical system
\begin{subequations}\label{eq:nominalSystem}
\begin{align}
\dot{x} &= f(x,u)\\
y &= g(x)
\end{align}
\end{subequations}
with $f:\mathbb{R}^n\times \mathbb{R}^m\rightarrow \mathbb{R}^n$ and $g:\mathbb{R}^n\rightarrow \mathbb{R}^m$ both continuously differentiable. Let $(x^*,0)\in \mathbb{R}^n\times \mathbb{R}^m$ be an equilibrium point, i.e., 
$$0 = f(x^*,0)$$
such that
$$0 = g(x^*).$$

Now introduce linear dynamic feedback of the form
\begin{subequations}\label{eq:linearCompensator}
\begin{align}
\dot{\xi} &= E \xi + Fy\\
u &= G\xi + Hy
\end{align}
\end{subequations}
with state $\xi\in\mathbb{R}^p$.

The overall closed-loop dynamical system becomes
\begin{subequations}\label{eq:ofStabilization}
\begin{align}
\dot{x} &= f(x,G\xi + Hg(x))\\
\dot{\xi} &= E\xi + Fg(x)
\end{align}
which has an equilibrium $(x^*,0)\in \mathbb{R}^n\times \mathbb{R}^p$. 
\end{subequations}

Our goal is to provide conditions on the feedback matrices $(E,F,G,H)$ so that the equilibrium $(x^*,0)$ of (\ref{eq:ofStabilization}) is locally asymptotically stable. Towards this end, define the matrices
\begin{align*}
A &= \nabla_x f(x^*,0)  \in \mathbb{R}^n\times \mathbb{R}^n\\
B &= \nabla_u f(x^*,0)  \in \mathbb{R}^n\times \mathbb{R}^m\\
C &= \nabla_x g(x^*)    \in \mathbb{R}^m\times \mathbb{R}^n
\end{align*}
and define $\tilde{x} = x - x^*$. The \textit{linearized dynamics} of (\ref{eq:ofStabilization}) at the equilibrium point $(x^*,0)$ are
$$\Pmatrix{\dot{\tilde{x}}\\ \dot{\xi}} = 
\Pmatrix{A + BHC & BG\\FC& E}\Pmatrix{\tilde{x}\\ \xi}.$$

\begin{Theorem}\label{thm:ofStabilization} If the matrix
$$\bar{A} = \Pmatrix{A + BHC & BG\\FC& E}$$
is a stability matrix, i.e., $\mathbf{Re}[\lambda_i] < 0$ for every eigenvalue, $\lambda_i$, then
the dynamics (\ref{eq:ofStabilization}) are locally exponentially stable at the equilibrium point $(x^*,0)$.
\end{Theorem}

This result is standard in the analysis of nonlinear systems, e.g., \cite[Theorem 4.7]{khalil2002nonlinear}.

\begin{Theorem}\label{thm:robustStability} Under the conditions of Theorem~\ref{thm:ofStabilization}, there exists a $\gamma > 0$ such that
$$\bar{A}' = \Pmatrix{A' + B'HC' & B'G\\FC'& E}$$
is a stability matrix for all $A'$, $B'$, and $C'$, such that
$$\norm{A - A'} < \gamma, \quad \norm{B - B'}<\gamma, \und \norm{C-C'} < \gamma.$$
\end{Theorem}

Again, the result is standard, e.g., \cite[Lemma 9.1]{khalil2002nonlinear}. The implication here is that if the nominal system (\ref{eq:nominalSystem}) is perturbed, albeit with a possibly shifted equilibrium $(x',0)$, then the \textit{same} feedback (\ref{eq:linearCompensator}) results in $(x',0)$ being asymptotically stable for the perturbed system as long as the perturbations are sufficiently small.

\subsection{Linear systems: Stabilizability and strong stabilizability}\label{sec:linearSystems}

The material in this section are standard and can be found in textbooks such as \cite{hespanha18linear} and \cite{rugh1996linear}.

Consider the linear system: 
\begin{subequations}\label{eq:linear}
\begin{align}
\dot{x}&=Ax+Bu,\quad x(0)=x_o\\
y&=Cx 
\end{align}
\end{subequations}
where $x \in \mathbb{R}^n,$ $u\in \mathbb{R}^m$, and $y \in \mathbb{R}^p$.  The system is called stable if the matrix $A$ is a stability matrix. 

The linear dynamic feedback 
\begin{subequations}\label{eq:efghcontroller}
\begin{align}
\dot{\xi}&=E\xi+Fy\\
u&=G\xi+Hy
\end{align}
\end{subequations}
stabilizes \eqref{eq:linear} if the overall system 
\begin{subequations}\label{eq:linearoverall}
\begin{align}
\Pmatrix{ \dot{x} \\ \dot{\xi}}=\Pmatrix{A+BHC & BG\\FC & E} \Pmatrix{x\\ \xi}
\end{align}
\end{subequations}
is stable. 

The following conditions are necessary and sufficient for the existence of feedback stabilizing dynamics for \eqref{eq:linear}:
\begin{itemize}
\item The pair $(A,B)$ is stabilizabile:
$$\Pmatrix{A-\lambda I & B}$$
has full row rank for all $\lambda\in \mathbb{C}$ with $\mathbf{Re}[\lambda] \geq 0.$
\item The pair $(A,C)$ is detectable: 
$$\Pmatrix{A-\lambda I \\ C}$$
has full column rank for all $\lambda\in \mathbb{C}$ with $\mathbf{Re}[\lambda] \geq 0.$
\end{itemize}

The matrix
$$\mathcal{T}(s)=C(sI-A)^{-1}B$$
is called the transfer matrix of \eqref{eq:linear} and is used to study the input-output properties of the system. Poles and zeros of a system are defined as follows (several definitions of zeros exist, we adopt the one relevant to our context): 
\begin{itemize}
\item A \textit{zero} is a value $\sigma \in  \mathbb{C}$ such that $\mathcal{T}(\sigma)=0.$
\item A \textit{pole} is a value $\mu \in \mathbb{C}$ such that for some $i$, $j$, 
$\magn{\mathcal{T}_{ij}(\mu)}=\infty,$
where $\mathcal{T}_{ij}(s)$ is the entry at the $i^{\text{th}}$ row and $j^{\text{th}}$ column of $\mathcal{T}$. The multiplicity of a pole is the largest multiplicity it has as a pole of  any minor of $\mathcal{T}(s).$ 
\end{itemize}

A stabilizable and detectable linear system of the form \eqref{eq:linear} is called \textit{strongly stabilizable} if there exist stabilizing dynamics of the form \eqref{eq:efghcontroller} with $E$ being a stability matrix.  The ``parity interlacing principle" \cite{youla1974singleloop} provides a necessary and sufficient condition for strong stabilizability. The principle  states that there should be an \emph{even} number of real poles, counting with multiplicities, between each pair of real unstable blocking zeros. 

\subsection{Washout filters}\label{App:WashoutFilters}

Washout filters are linear systems that take the form
\begin{align*}
\dot{x}&=-x+u\\
y&=-x+u,
\end{align*}
where $x\in\mathbb{R}^n$, $u\in\mathbb{R}^n$, and $y\in\mathbb{R}^n$. 
These systems have the property that if 
$$\lim_{t\rightarrow \infty} u(t)= u^*, \quad \text{then} \quad \lim_{t\rightarrow \infty}y(t)= 0.$$
For an extended discussion, see \cite{hassouneh2004washout}. 

\subsection{Decentralized stabilization} \label{App:Decentralized}

 Consider the linear system 
\begin{subequations}\label{eq:linear1}
\begin{align}
\dot{x}&=Ax+\sum_{i=1}^{k} B_i u_i, \\  y_i&=C_ix \quad \quad  i=1,\hdots,k ,
\end{align}
\end{subequations}
with $k$ inputs and $k$ outputs, and $A\in \mathbb{R}^{n\times n}$.
If there exist $(E_i,F_i,G_i,H_i)$, $i=1,\hdots,k$, such that the collective system
\begin{align*}
\dot{x}&=Ax+\Pmatrix{B_1 & \hdots & B_k}\Pmatrix{G_1 \xi_1 + H_1 y_1 \\ \vdots \\ G_k \xi_k + H_k y_k}\\
\dot{\xi}_i & = E_i \xi_i + F_i y_i, \quad i=1,\hdots, k 
\end{align*}
is stable, then the system is stabilized in a decentralized manner. 

A necessary and sufficient condition for decentralized stabilizability is presented in \cite{davison1990decentralized}. First, for any partition $\mathcal{U}\cup \mathcal{Y} = \theset{1,2,...,k}$
define $B\vert^\mathcal{U}$ as the matrix formed by all $B_i $ with indices in $\mathcal{U}$, i.e.,
$$B\vert^\mathcal{U} = \Pmatrix{B_{q_1}&...&B_{q_\magn{\mathcal{U}}}}$$
with $\theset{q_1,...,q_\magn{\mathcal{U}}} = \mathcal{U}$. Likewise, define $C\vert_\mathcal{Y}$ as the matrix 
formed by all $C_i$ with indices in $\mathcal{Y}$, i.e.,
$$C\vert_\mathcal{Y} = \Pmatrix{C_{e_1}\\ \vdots\\ C_{e_\magn{\mathcal{Y}}}}$$
with $\theset{e_1,...,e_\magn{\mathcal{Y}}} = \mathcal{Y}$.

\begin{Theorem}[\cite{davison1990decentralized}, Theorem 3]\label{thm:davison} 
System \eqref{eq:linear1} is stabilizable in a decentralized manner if and only if
$$\mathbf{rank}\Pmatrix{A - \lambda I&B\vert^\mathcal{U}\\ C\vert_\mathcal{Y}&0} = n$$
for all complex $\lambda$ with $\mathbf{Re}[\lambda]\ge 0$ and all partitions, $\mathcal{U}\cup \mathcal{Y} = \theset{1,2,...,k}$.
\end{Theorem}

\subsection{ODE method of stochastic approximation}\label{sec:ODEMethod}

The following proposition serves as the main tool for proving our results in the bandit setting.  

\begin{Proposition}[\cite{Benaim1999Dynamics}]\label{Prop:SI}
Consider the following stochastic iterates
\begin{equation}\label{eq:SI}
\zeta(t+1)=\zeta(t)+\kappa(t)\left[F(\zeta(t))+H(t)\right]
\end{equation}
where the following conditions hold:
\begin{enumerate}
\item The step-sizes $\kappa(t) > 0$ satisfy:
$$\sum_{t=1}^{\infty} \kappa(t) = \infty,\quad \sum_{t=1}^{\infty} \kappa(t)^2 < \infty.$$

\item The iterates evolve in a compact set $K.$

\item $F(\zeta)$ is uniformly Lipschitz over $K$:
$$\exists L \quad s.t \quad \norm{F(\zeta) - F(\zeta')}  \leq L \norm{\zeta-\zeta'} \text{ } \forall \zeta,\zeta' \in K.$$

\item $H(t)$ is a stochastic noise such that 
$$\mathbf{E}[H(t)\vert \zeta(t)]=0 \quad \text{and} \quad \mathbf{E}[H(t)\tr H(t)]<\infty.$$

\item The equilibrium point $x^*\in K$ satisfying $F(x^*)=0$ is attainable, i.e., for each $t>0$ and every open neighborhood $U$ of $x^*$, 
$\mathbf{Prob}[\exists s \geq t : \zeta(s) \in U]>0$\label{enum:A6}.

\end{enumerate}
Under these conditions, if $x^*$ is a locally asymptotically stable fixed point of the mean dynamics
\begin{equation}\label{eq:Mean} 
\dot{x}= F(x),
\end{equation}
 then 
 $$\mathbf{Prob} \left[\lim_{t\rightarrow \infty} \zeta(t) = x^*\right]>0.$$
\end{Proposition}

If $x^*$ is a globally asymptotically stable fixed point of \eqref{eq:Mean}, the result in the above proposition strengthens to almost sure convergence.

\bibliographystyle{ieeetr}
\bibliography{AllcitationsNOURL}
\end{document}